%% file: main.tex
\documentclass[12pt]{article} 

\usepackage[margin=1in]{geometry} 

\usepackage[sectionbib]{natbib}
\usepackage{array,epsfig,fancyheadings,rotating}
\usepackage[]{hyperref}  

\usepackage{titlesec}
\usepackage{sectsty, secdot}
\usepackage{setspace}

\titlespacing*{\section}    {0pt}{1.4ex plus .2ex minus .2ex}{0.8ex plus .2ex}
\titlespacing*{\subsection} {0pt}{1.2ex plus .2ex minus .2ex}{0.6ex plus .2ex}
\titlespacing*{\subsubsection}{0pt}{1.0ex plus .2ex minus .2ex}{0.5ex plus .2ex}

\usepackage{amsmath}
\usepackage{amssymb}
\usepackage{amsfonts}
\usepackage{multirow}
\usepackage{amsthm}

\newtheorem{theorem}{Theorem}
\newtheorem{lemma}{Lemma}

\newtheorem{proposition}{Proposition}
\theoremstyle{definition}

\usepackage{booktabs}

\newtheorem{assumption}{Assumption}
\usepackage{mathtools}
\usepackage{bm}
\usepackage{multirow}
\usepackage{float}
\usepackage{amsmath}
\usepackage{amssymb}
\usepackage{amsfonts}
\usepackage{multirow}
\usepackage{amsthm}
\usepackage{caption}

\usepackage{enumitem}

\input{notations}

\usepackage{xcolor}
\definecolor{vz}{rgb}{0.0, 0.0, 0.7}

\allowdisplaybreaks

\DeclareCaptionFont{tablecaptionstretch}{\setstretch{1.2}}
\usepackage{xr}
\begin{document}


\renewcommand{\baselinestretch}{1}
\renewcommand{\arraystretch}{1.5}
\captionsetup[table]{font=tablecaptionstretch}
%


\title{\bf Longitudinal Adaptive Experimental Design for Learning Multiple Target Estimands with Semiparametric Efficient Inference}
  \author{Wenxin Zhang, Mark van der Laan \thanks{
    \looseness -1
    This work was supported by the National Institutes of Health under Grant R01AI074345. The content is solely the responsibility of the authors and does not necessarily represent the official views of the NIH. 
    } \\
    Division of Biostatistics, University of California, Berkeley}
    \maketitle
\begin{abstract}
Adaptive designs are increasingly used in clinical trials and digital experiments to improve estimation efficiency by updating treatment randomization probabilities as data accumulate. While most existing work focuses on settings with a single-stage treatment, adaptive designs for longitudinal studies with multi-stage, time-varying treatments remain relatively underexplored. In this work, we develop a general semiparametric efficiency framework for designing longitudinal adaptive experiments to optimize the estimation efficiency of a broad class of target estimands of interest. An efficiency-oriented design criterion is proposed to accommodate both single-estimand targets and joint optimization across multiple estimands. We demonstrate that optimal randomization at earlier stages depends on later-stage allocations, yielding a backward-recursive strategy for deriving the oracle design, and propose a longitudinal adaptive design to sequentially learn and target the oracle design using accumulating data. 
We further develop an adaptive-design-likelihood-based longitudinal targeted maximum likelihood estimator (ADL-LTMLE) for asymptotically
normal and semiparametric efficient estimation of statistical estimands from dependent data collected from adaptive experiments, without relying on parametric model assumptions.
Applying the framework to time-to-treatment-initiation effects that compare initiating treatment at a given stage with delaying initiation until a subsequent stage, we show that designs optimized for a particular stage-specific effect can substantially compromise estimation efficiency for effects defined at other stages, highlighting the design trade-offs addressed by our framework. Simulation studies show the proposed design and estimation approaches achieve substantial variance reductions relative to non-adaptive designs, with performance close to that of the oracle design.

\vspace{9pt}
\noindent {\it Key words and phrases: covariate-adjusted response-adaptive design, longitudinal data, semiparametric efficiency, targeted maximum likelihood estimation}
\par
\end{abstract}\par

\def\thefigure{\arabic{figure}}
\def\thetable{\arabic{table}}

\fontsize{12}{14pt plus.8pt minus .6pt}\selectfont


\setstretch{1.7}
\section{Introduction}

Adaptive designs are increasingly used in clinical trials and digital experiments \citep{chow2008adaptive,robertson2023response,larsen2024statistical}, allowing treatment randomization probabilities to adapt in response to accumulating data.
Dating back to \cite{thompson1933likelihood},
a central class of such methods is Response-Adaptive Randomization (RAR), in which treatment assignment probabilities are updated using treatment and outcome information from previously enrolled participants, with objectives ranging from improving efficiency to reducing assignment to inferior treatments \citep{wei1978randomized,hu2004asymptotic,hu2006theory,robertson2023response}. Covariate-Adjusted Response-Adaptive (CARA) designs generalize RAR by incorporating baseline covariates into the randomization mechanism \citep{rosenberger2001covariate,hu2006theory,zhang2007asymptotic,van2008construction,rosenberger2008handling,atkinson2011covariate,hu2015unified}. 
These adaptive designs are also related to developments in the multi-armed and contextual bandit literature, which studies sequential decision-making strategies with the objective of regret minimization \citep{bubeck2012regret,lattimore2020bandit}.

However, relatively few studies have investigated adaptive experimental designs in longitudinal settings. Existing work in this area primarily primarily considers settings with longitudinally observed responses or covariates but a single treatment decision, where intermediate information is used to adapt the randomization for that decision 
\citep{huang2013longitudinal,huang2016longitudinal,biswas2004randomized,huang2009using,yang2024targeting,zhang2024evaluating,gao2024response}. Nevertheless, these approaches do not consider adaptive designs in more complex longitudinal settings involving multi-stage, time-varying treatments. Such settings are canonical in the causal inference literature \citep{robins1986new,murphy2003optimal,bang2005doubly,van2011LTMLE,van2018targeted} and arise commonly in clinical
decision-making, precision medicine, and sequential multiple assignment
randomized trials
\citep{lavori2000design,murphy2005experimental,lei2012smart,
kosorok2019precision,kosorok2021introduction}.
In these scenarios, participants are followed over multiple stages, intermediate covariates or responses are revealed sequentially, and treatment at each stage is assigned based on a participant’s evolving history until the final outcome is observed.
This longitudinal structure fundamentally differs from experiments with a single treatment stage and calls for adaptive randomization strategies that explicitly account for the sequential nature of treatment assignment and information accrual.
While prior work has considered adaptive experiments or bandit algorithms with time-varying treatments to improve participants’ outcomes \citep{van2008construction,hu2020dtr,wang2022adaptive,ghosh2024optimal}, principled approaches to optimizing the estimation efficiency of statistical estimands remain relatively underdeveloped in these longitudinal settings.

In this paper, we develop a semiparametric efficiency framework for constructing longitudinal adaptive experiments with stage-wise treatment randomization mechanisms that efficiently learn one or multiple target
estimands. We further propose a general framework for semiparametric efficient estimation of statistical estimands from 
dependent data collected under longitudinal adaptive experiments, which applies to both the efficiency-oriented designs in this work and broader classes of adaptive experiments.

The contributions of this work are as follows.
(i) We propose a general efficiency-oriented design criterion for longitudinal experiments that targets the estimation efficiency of one or multiple estimands within a unified formulation.
(ii) For a broad class of longitudinal target estimands, we characterize the oracle design that optimizes the proposed criterion and uncover its forward-looking structure: optimal randomization at earlier stages depends on optimal allocations at later stages. This structure yields a backward-recursive strategy for deriving explicit stage-specific oracle randomization probabilities. 
(iii) We develop a covariate-adjusted response-adaptive (CARA) design that sequentially updates stage-wise randomization probabilities by learning and targeting the oracle design using accumulating data.
(iv) Finally, we propose an adaptive-design-likelihood longitudinal targeted maximum likelihood estimator (ADL-LTMLE) that enables asymptotically normal and efficient estimation of target estimands in the multi-stage adaptive experiments with time-varying treatments. Our approach builds on the literature on targeted maximum likelihood estimation (TMLE) for adaptively collected data \citep{van2008construction,chambaz2011targeting,chambaz2017targeted,malenica2021adaptive,malenica2024adaptive,zhang2024evaluating}, and recent work that establishes semiparametric efficiency of the ADL-TMLE approach under adaptive experiments \citep{zhang2025efficient}.

\looseness=-1
We illustrate the framework through longitudinal adaptive experiments for estimating multiple treatment-initiation effects. Simulations show that the proposed adaptive design yields substantial variance reduction relative to non-adaptive designs, especially when combined with the efficient ADL-LTMLE approach, with its performance close to that of the oracle design.

Overall, this work contributes a general maximal-information-type framework
for the design and analysis of adaptive experiments with longitudinal,
time-varying treatments. The framework accommodates a wide range of target
estimands and their joint collections, including estimands of causal effects
and mean outcomes under dynamic or stochastic treatment rules, and recovers
single-stage adaptive designs as a special case.

The remainder of the paper is organized as follows.
Section~\ref{sec:setup} describes the statistical setup.
Section~\ref{sec:general_framework} introduces the proposed
efficiency-oriented design criterion, and develops a backward-recursive
strategy for characterizing oracle designs for either a single target
estimand or multiple target estimands and constructing the corresponding
longitudinal adaptive designs.
Section~\ref{section:adl-ltmle} develops semiparametric efficient
adaptive-design-likelihood-based longitudinal TMLE methods for estimating
target statistical estimands and providing inference based on data collected
from longitudinal adaptive experiments.
Section~\ref{sec:example_estimand} presents treatment-initiation effects
as an illustrative application of the general framework and derives oracle
and adaptive designs for their efficient joint estimation.
Section~\ref{sec:simulation} reports the simulation studies.
Section~\ref{sec:conclusion} concludes with a discussion.
\section{Statistical Setup}
\label{sec:setup}

\subsection{Data Structure}
\looseness=-1
Consider a longitudinal data structure
$
O = (L_1, A_1, \ldots, L_K, A_K, L_{K+1} \equiv Y)$
over $K+1$ follow-up stages for each participant. At each stage $k \in [K] := \{1,\ldots,K\}$, 
$L_k \in \mathcal L_k$ denotes covariates (which can also be intermediate responses for $k > 1$), $A_k \in \mathcal{A}_k$ denotes the categorical treatment assigned at stage $k$. The final outcome $Y \in \mathcal Y$, also denoted by $L_{K+1}$, is a bounded outcome observed at stage $K+1$, which we assume without loss of generality that $\mathcal Y = [0,1]$. Let $\bar A_k := (A_1,\ldots,A_k)$ and $\bar L_k := (L_1,\ldots,L_k)$ denote the treatment and covariate histories up to stage $k$, and define $H_k := (\bar A_{k-1}, \bar L_k) \in \mathcal{H}_k$ as the participant's observed history immediately before treatment $A_k$.
Let $Q_1$ denote the marginal distribution of baseline covariates $L_1$, and for $k = 1,\ldots,K$, let $Q_{k+1}$ denote the conditional distribution of $L_{k+1}$ given $(A_k,H_k)$. We collect these distributions into $\bfQ := (Q_1,\ldots,Q_{K+1})$. Let $\bfg := (g_1,\ldots,g_K)$ denote a collection of stage-wise treatment randomization mechanisms, where each $g_k: \cH_k \times \cA_k \to [0,1]$ maps a history-treatment pair
$(h_k,a_k)$ to the conditional probability of assigning treatment $a_k$
given history $h_k$. We write $g_k(a_k \mid h_k)$ as shorthand for
$P(A_k = a_k \mid H_k = h_k)$ under the randomization mechanism $g_k$.
We consider the nonparametric statistical model
$
\mathcal M
=
\{ P_{\bfQ,\bfg} : \bfQ \in \mathcal Q,\ \bfg \in \mathcal G \}$,
where each $P_{\bfQ,\bfg}$ admits the following likelihood factorization:
\begin{align}
p_{\bfQ,\bfg}(o)
=
q_1(l_1)
\prod_{k=1}^K
g_k(a_k \mid h_k)\,
q_{k+1}(l_{k+1} \mid a_k, h_k), \label{eq:likelihood}
\end{align}
which defines the joint density of the longitudinal data $O$ under
$P_{\bfQ,\bfg}$.
Here, $\mathcal Q$ denotes the nonparametric model for
$\bfQ=(Q_1,\ldots,Q_{K+1})$, and
$\mathcal G$ denotes the class of admissible treatment 
randomization functions.

\subsection{Longitudinal Adaptive Experiment}
\label{sec:longitudinal_adaptive_design}

Consider a sequential longitudinal adaptive experiment with participants
indexed by $i=1,\ldots,n$, where each participant is followed for $K+1$
stages. For each participant $i$, we observe longitudinal data
$
O_i
=
(L_{1,i},A_{1,i},\ldots,L_{K,i},A_{K,i},Y_i),
$
where $L_{k,i}$ denotes time-varying covariates at stage $k$,
$A_{k,i}\in \cA_{k}$ denotes the categorical treatment assigned at stage $k$, and
$Y_i\equiv L_{K+1,i}$ is the final outcome.
Define the accumulated data up to participant $i$ by
$
\bO_i
:=
\{O_1,\ldots,O_i\}$, where $\bO_0:=\emptyset$.
Let $\mathcal F_i
=
\sigma(O_1,\ldots,O_i)$
denote the filtration generated by the sequentially observed data.

Before participant $i$ is enrolled, the experimenter uses the accumulated
data $\bO_{i-1}$ to construct a collection of stage-wise randomization
mechanisms
$
\bfg_i
=
(g_{1,i},\ldots,g_{K,i})$.
Conditional on $\bO_{i-1}$, treatment at stage $k$ for participant $i$ is
assigned according to the known randomization mechanism
$g_{k,i}(\cdot\mid H_{k,i})$. These mechanisms govern the participant's
treatment assignments sequentially across stages throughout longitudinal
follow-up.
This formulation also includes cohort-wise adaptation as
a special case. Specifically, if participants are enrolled in cohorts, the
same collection of randomization mechanisms $\bfg_i$ may be applied to all
participants within a cohort and updated only between cohorts.

Let
$
\bfg_{1:n}
:=
(\bfg_1,\ldots,\bfg_n)$
denote the known adaptive design sequence applied to the $n$ participants.
For a given $(\bfQ,\bfg_{1:n})$, we write
$P_{\bfQ,\bfg_{1:n}}$ for the joint distribution of
$(O_1,\ldots,O_n)$. The likelihood of $(o_1,\ldots,o_n)$ is given by
\begin{eqnarray*}
p_{\bfQ,\bfg_{1:n}}(o_1,\ldots,o_n)
&=&
\prod_{i=1}^n
\left[
q_1(l_{1,i})
\prod_{k=1}^K
\Bigl\{
g_{k,i}(a_{k,i}\mid h_{k,i})\,
q_{k+1}(l_{k+1,i}\mid a_{k,i},h_{k,i})
\Bigr\}
\right].
\end{eqnarray*}

\section{A General Framework for Efficiency-Oriented Longitudinal Adaptive Design}
\label{sec:general_framework}

\subsection{Oracle Longitudinal Design for Estimation Efficiency}
\label{subsec:design_criterion}
\looseness=-1
Let $\psi(\bfQ)$ denote a pathwise differentiable target estimand of
interest. We introduce a semiparametric efficiency framework for defining
an efficiency-oriented design criterion. Specifically, for a fixed design $\bfg$, define
\begin{equation}
C(\bfQ,\bfg) := \Var_{P_{\bfQ,\bfg}} \left\{ D(P_{\bfQ,\bfg})(O) \right\},
\label{eq:design_objective}
\end{equation}
where $D(P_{\bfQ,\bfg})$ denotes the efficient influence curve (EIC) of
$\psi(\bfQ)$ under the observed-data distribution $P_{\bfQ,\bfg}$.
Under standard regularity conditions, the pathwise derivative of
$\psi(\bfQ)$ along any regular parametric submodel can be represented as
the inner product of its score with an influence function. By the Riesz
representation theorem, the EIC is the unique element of the tangent space
that represents this derivative and has the smallest $L^2(P_{\bfQ,\bfg})$
norm among all influence functions. Consequently,
\eqref{eq:design_objective} is the semiparametric efficiency bound that
characterizes the smallest asymptotic variance attainable by a regular
asymptotically linear estimator of $\psi(\bfQ)$
\citep{laan2003unified}.

While the collection of randomization mechanisms
$\bfg=(g_1,\ldots,g_K)$ does not alter the definition of $\psi(\bfQ)$,
it determines the induced observed-data distribution
$P_{\bfQ,\bfg}$ and, consequently, the statistical information available
for estimating $\psi(\bfQ)$. This motivates the oracle maximal-information longitudinal design,
defined as
\begin{equation}
\bfg^{\star}(\bfQ)
:=
\arg\min_{\bfg\in\mathcal G}
C(\bfQ,\bfg).
\label{eq:oracle_design_def_single}
\end{equation}

For many longitudinal target estimands, the EIC admits a stage-wise
decomposition
$
D(P_{\bfQ,\bfg})
=
\sum_{k=0}^K D_k(P_{\bfQ,\bfg})$,
where $D_0(P_{\bfQ,\bfg})$ is the mean-zero baseline component associated
with the marginal law of $L_1$. For each $k\in[K]$,
$D_k(P_{\bfQ,\bfg})(O)$ is the component associated with the conditional
law of $L_{k+1}$ given $(H_k,A_k)$; it is measurable with respect to
$(H_k,A_k,L_{k+1})$ and satisfies
$
E_{P_{\bfQ,\bfg}}
\left\{
D_k(P_{\bfQ,\bfg})(O)
\mid H_k,A_k
\right\}
=
0$.
These stage-specific components form martingale increments and are
mutually orthogonal in $L^2(P_{\bfQ,\bfg})$. Consequently, the variance
of the EIC decomposes as
$
\Var_{P_{\bfQ,\bfg}}
\left\{
D(P_{\bfQ,\bfg})(O)
\right\}
=
\sum_{k=0}^K
E_{P_{\bfQ,\bfg}}
\left[
\left\{
D_k(P_{\bfQ,\bfg})(O)
\right\}^2
\right].
$
Thus, the semiparametric efficiency bound separates into stage-specific
contributions through which the longitudinal treatment mechanisms affect
the total estimation variance.

\subsection{Backward-Recursive Structure of the Oracle Design}
\label{subsec:general_recursive}
\label{sec:oracle_design}

We consider a broad class of longitudinal target estimands whose  
stage-specific EIC components have the following representation:
\begin{equation}
D_k(P_{\bfQ,\bfg})(O) = \frac{ \Phi_k(\bfQ)(O) }{ \prod_{s=1}^k g_s(A_s\mid H_s) }, \qquad k = 1,\ldots,K,
\label{eq:general_estimand}
\end{equation}
where we require that each $\Phi_k(\bfQ)$ is an estimand-specific function that does not depend on the experimental mechanism $\bfg$. 

\looseness=-1
This representation characterizes a general class of statistical estimands supported by the proposed design optimization framework, including many classical causal estimands in longitudinal settings such as mean outcomes under static, dynamic, or stochastic treatment rules and their linear combinations (see, e.g., \cite{robins1986new,laan2003unified,murphy2003optimal,van2015targeted,wang2025targeted}).

Importantly, for this class of estimands, the stage-wise EIC components exhibit a nested dependence on the longitudinal treatment mechanisms, as indicated by the following simplified indexing:
\[
D_1(\Phi_1,g_1),\quad
D_2(\Phi_2,g_1,g_2),\quad
\ldots,\quad
D_K(\Phi_K,g_1,\ldots,g_K).
\]
Specifically, conditional on the realized history $H_k$ available before treatment assignment at stage $k$, the randomization mechanism $g_k$ affects the variance contribution of $D_k$ and all subsequent components $D_{k+1},\ldots,D_K$, where the latter contributions further depend on the future treatment mechanisms $g_{k+1},\ldots,g_K$. 

This nested dependence gives rise to a forward-looking, backward-recursive optimization of the longitudinal design. The recursion starts by optimizing the terminal randomization mechanism $g_K$ and then proceeds backward through
$g_{K-1},\ldots,g_1$, so that each earlier-stage randomization mechanism accounts for its downstream
contribution to the variance through the subsequent data trajectories it induces. 
The procedure is as follows.

First, at the final stage $K$, let
$
\sigma_K^2(a_K\mid h_K) = E \left[ \left\{ \Phi_K(\bfQ)(O) \right\}^2 \mid H_K=h_K,A_K=a_K \right]$.
Conditional on $H_K=h_K$, minimizing the final-stage variance contribution of $D_K$ yields the allocation:
\begin{equation}
g_K^\star(a_K\mid h_K) = \frac{ \sigma_K(a_K\mid h_K) }{ \displaystyle \sum_{a_K'\in\mathcal A_K} \sigma_K(a_K'\mid h_K) }.
\label{eq:oracle_terminal}
\end{equation}
Intuitively, the oracle design assigns
greater probability to treatment arms associated with larger current and future uncertainty conditional on the current history. Collecting more observations under these arms reduces their contribution to the overall uncertainty and improves the efficiency of estimating the target estimand.

Proceeding backward, for \(k=K-1,\ldots,1\), define
\begin{align}
\sigma_k^2(a_k\mid h_k)
={}&
E\left[
\left\{
\Phi_k(\bfQ)(O)
\right\}^2
\mid
H_k=h_k,A_k=a_k
\right]
\nonumber\\
&+
E\left[
\left.
\sum_{a_{k+1}\in\mathcal A_{k+1}}
\frac{
\sigma_{k+1}^2(a_{k+1}\mid H_{k+1})
}{
g_{k+1}^\star(a_{k+1}\mid H_{k+1})
}
\right|
H_k=h_k,A_k=a_k
\right],
\label{eq:recursive_sigma_single}
\end{align}
which aggregates the conditional
variance contributions of \(D_k,\ldots,D_K\) under the
future oracle mechanisms
\((g_{k+1}^\star,\ldots,g_K^\star)\).

The resulting stage-\(k\) oracle allocation is given by:
\begin{equation}
g_k^\star(a_k\mid h_k)
=
\frac{
\sigma_k(a_k\mid h_k)
}{
\displaystyle
\sum_{a_k'\in\mathcal A_k}
\sigma_k(a_k'\mid h_k)
},
\qquad k=K,\ldots,1.
\label{eq:oracle_single_recursive}
\end{equation}

The final collection 
$
\bfg^\star
:=
\left(
g_1^\star,\ldots,g_K^\star
\right)$ thus 
constitutes the oracle longitudinal design across all stages.

\subsection{Oracle Longitudinal Design for Multiple Target Estimands}
\label{subsec:joint_optimization}
With rich data trajectories involving time-varying treatments and responses, longitudinal studies create opportunities to address a wide range of scientific questions through multiple target estimands of interest. However, they also give rise to complex trade-offs in experimental design, as a design optimized for one estimand may provide less favorable data support for estimating others. This challenge is particularly substantial in longitudinal experiments, where efficiency losses due to limited support may propagate multiplicatively through the stage-wise dependence of the longitudinal data.
This motivates a joint design criterion that explicitly accounts for the 
trade-offs in estimation efficiency across multiple target estimands.

\looseness=-1
Consider a collection of target estimands of interest 
$\{\psi^j(\bfQ):j\in[J]\}$. We define the joint
efficiency-oriented criterion
\begin{equation}
C_w(\bfQ,\bfg)
:=
\sum_{j\in[J]}
w^j
\Var_{P_{\bfQ,\bfg}}
\left\{
D^j(P_{\bfQ,\bfg})(O)
\right\},
\label{eq:design_objective_multiple}
\end{equation}
where $D^j(P_{\bfQ,\bfg})$ denotes the EIC corresponding to
$\psi^j(\bfQ)$ under $P_{\bfQ,\bfg}$, and the user-specified
nonnegative weights $\{w^j:j\in[J]\}$ satisfying $\sum_{j=1}^J w^j = 1$ determine the relative emphasis
placed on the estimation efficiency of each target estimand.

When the EIC of each target estimand admits the stage-wise representation in
\eqref{eq:general_estimand}, the optimization for the joint criterion also preserves the same
backward-recursive structure. Specifically, for each estimand \(\psi^j\), let
\(\sigma_k^{2,j}(a_k\mid h_k)\) denote the conditional second-moment
contribution of its stage-\(k\) and subsequent EIC components under the
common future oracle mechanisms. Define
$
\Sigma_k^2(a_k\mid h_k)
=
\sum_{j\in[J]}
w^j
\sigma_k^{2,j}(a_k\mid h_k)$.
The oracle allocation conditional on the history of stage $k$ is then
$
g_k^\star(a_k\mid h_k)
=
\frac{
\Sigma_k(a_k\mid h_k)
}{
\displaystyle
\sum_{a_k'\in\mathcal A_k}
\Sigma_k(a_k'\mid h_k)
}$
for $k=K,\ldots,1$.
Thus, the joint criterion retains the same stage-wise allocation form,
while aggregating uncertainty across the target estimands according to
user-specified weights that reflect their relative estimation priorities.

\looseness=-1
More generally, let
$
\boldsymbol D(P_{\bfQ,\bfg})
=
\left\{
D^1(P_{\bfQ,\bfg}),
\ldots,
D^J(P_{\bfQ,\bfg})
\right\}^{\top}$
denote the vector-valued EIC, and consider 
$
C_{\boldsymbol\Omega}(\bfQ,\bfg)
=
\operatorname{tr}
\left[
\boldsymbol\Omega\,
\Var_{P_{\bfQ,\bfg}}
\left\{
\boldsymbol D(P_{\bfQ,\bfg})(O)
\right\}
\right]$ as a design criterion,
where
$
\Var_{P_{\bfQ,\bfg}}
\left\{
\boldsymbol D(P_{\bfQ,\bfg})(O)
\right\}$
is the covariance matrix of the EICs and
$\boldsymbol\Omega$ is a user-specified symmetric positive-semidefinite
matrix. In particular, let $\boldsymbol\Omega
:=
B^\top\Lambda B$, where $B \in \mathbb{R}^{R\times J}$ is a prespecified
contrast matrix whose rows define $R$ linear combinations of the $J$
target estimands, and let 
$
\Lambda
=
\operatorname{diag}(\lambda^1,\ldots,\lambda^R)$
be a diagonal matrix of nonnegative weights. This
criterion represents a weighted sum of the semiparametric
efficiency bounds for the user-specified linear combinations of the estimands. The backward-recursive optimization procedure continues to apply.

\subsection{Efficiency-Oriented Longitudinal Adaptive Design}
However, oracle designs depend on unknown features of the data-generating
distribution and therefore cannot be implemented \textit{a priori}.
This motivates a longitudinal adaptive design that estimates the
relevant quantities sequentially from accumulating data
and updates the stage-specific randomization mechanisms accordingly.

Specifically, before enrolling participant $i$, we can use the accumulated observations 
$\bO_{i-1}$ to estimate
$\{\sigma_k^2(a_k\mid H_k):k\in[K]\}$ for a single target estimand, or
$\{\Sigma_k^2(a_k\mid H_k):k\in[K]\}$ for multiple target estimands.
Substituting these estimates into the backward-recursive solution yields
the estimated oracle mechanisms
$\{\hat g_{k,i}^\star(\cdot\mid H_k):k\in[K]\}$.
The stage-specific treatment mechanisms for participant $i$ are then
given by
$g_{k,i}(\cdot\mid H_k)
=
\hat g_{k,i}^\star(\cdot\mid H_k)$ for $k\in[K]$.
By doing so, the longitudinal adaptive design updates the randomization
mechanisms as information accumulates, thereby steering the distribution of the accruing data toward that under the oracle design with the intended efficiency objective.

\subsection{Inference After Longitudinal Adaptive Experiments}
While adaptive experiments are attractive because they can steer the randomization mechanisms toward the design objective, the resulting dependence across participants invalidates standard inferential procedures developed for i.i.d.\ data, particularly when
inference is conducted without parametric assumptions. In longitudinal experiments, statistical inference is further complicated by
the sequential dependence among time-varying covariates, treatments, and
outcomes within each participant. In the next section, we provide a targeted maximum likelihood estimation (TMLE) framework to address these challenges.

\section{Adaptive-Design-Likelihood-Based Longitudinal TMLE}
\label{section:adl-ltmle}
We develop an adaptive-design-likelihood-based longitudinal TMLE (ADL-LTMLE) approach for semiparametric efficient inference under longitudinal adaptive experiments, building on recent advances in adaptive-design-likelihood-based TMLE (ADL-TMLE) framework \citep{zhang2025efficient}. This framework relates semiparametric efficiency bound under an adaptive experiment to the corresponding bound under an i.i.d.\ reference distribution induced by the average of the distributions generated under the adaptive design. It thereby yields an asymptotically normal and semiparametric efficient estimator for adaptively collected data without relying on parametric model assumptions. 

This framework applies to a broad class of longitudinal target estimands,
including those defined under dynamic or stochastic treatment rules and other
estimands that admit analogous TMLE procedures based on martingale and
semiparametric efficiency theory. More generally, it applies to target
estimands of the form $\psi(\bfQ)$ for data likelihoods factorizing as
$
p(o_1,\ldots,o_n)=\prod_{i=1}^n \mathbf q(o_i)\bfg_i(o_i),
$
where $\mathbf q$ is the density of some $\bfQ\in\mathcal Q$ that defines
the target estimand. Thus, the framework is not restricted to a particular
type of adaptive experiment or to the estimands used to guide the adaptive
design, provided that the observed data offer sufficient support for the
target estimand of interest.

For the rest of this section, we use $\psi^{\bd}(\bfQ)$, the mean outcome under a dynamic treatment rule $\bd$, as a representative example of the proposed framework in longitudinal settings. We first formalize this estimand and introduce the proposed ADL-LTMLE procedure, and then establish its asymptotic normality and semiparametric efficiency, with proofs provided in
the Supplementary Material.
\subsection{Expected Outcome under a Dynamic Treatment Rule}
A dynamic treatment rule is a sequence of decision rules
$\bd=(d_1,\ldots,d_K)$, where each
$d_k:\mathcal H_k\to\mathcal A_k$ maps the observed history $H_k$
to a treatment at stage $k$. For each $k\in[K]$, let
$
L_{k+1:K+1}^{\bd}
=
\left(
L_{k+1}^{\bd},\ldots,L_{K+1}^{\bd}
\right)
$
denote the counterfactual trajectory of covariates/outcomes from stage
$k+1$ onward under rule $\bd$, with
$L_{K+1}^{\bd}=Y^{\bd}$. The causal estimand of interest is
$
\Psi^{\bd}
:=
E\left(Y^{\bd}\right),
$
the mean counterfactual outcome under rule $\bd$. Let $P_{\bfQ_0}^{\bd}$ denote the intervention distribution induced by
$\bfQ_0$ when treatment is assigned by rule $\bd$, and define
$
\bar G_{k,n}^{\bd}(h_k)
=
\frac{1}{n}\sum_{i=1}^n
\prod_{s=1}^k
g_{s,i}\{d_s(h_s)\mid h_s\}$
as the average cumulative probability of following $\bd$ across participants in the adaptive experiment.

\begin{assumption}[Identification]
\label{assumption:causal-identification-d}
For a treatment rule $\bd=(d_1,\ldots,d_K)$, suppose that:
\begin{enumerate}
\item[(i)] \textit{Consistency:} if
$A_{s,i}=d_s(H_{s,i})$ for every $s\leq k$, then
$L_{k+1,i}=L_{k+1,i}^{\bd}$;
\item[(ii)] \textit{Sequential randomization:} for every $k\in[K]$,
$
L_{k+1:K+1}^{\bd} 
\indep
A_{k,i}
\mid
H_{k,i},\bO_{i-1}$,
where $\bO_{i-1}$ denotes accumulated data before participant
$i$ is enrolled;
\item[(iii)] \textit{Positivity:} $\forall k\in[K]$,
$
\bar G_{k,n}^{\bd}(h_k)>0
$ for $P_{\bfQ_0}^{\bd}$-almost every history $h_k$.
\end{enumerate}
\end{assumption}

Here, consistency links the observed and counterfactual variables
along the target treatment rule, sequential randomization removes
treatment confounding conditional on the observed history and prior
experimental data, and positivity ensures that the adaptive experiment
provides support for the treatment trajectory required to evaluate
$\bd$.

To state the identification formula, we set
$m_{K+1}^{\bd}(\emptyset,H_{K+1}) \equiv Y$ throughout and, for
$k=K,\ldots,1$, define
\[
m_k^{\bd}(a_k,h_k)
=
E_{\bfQ}
\left[
m_{k+1}^{\bd}
\{d_{k+1}(H_{k+1}),H_{k+1}\}
\mid
A_k=a_k,H_k=h_k
\right],
\]
with the convention that $m_{k+1}^{\bd}
\{d_{k+1}(H_{k+1}),H_{k+1}\} \equiv Y$ when $k=K$. Under the conditions above, the longitudinal
$g$-computation formula \citep{robins1986new} identifies $\Psi^{\bd}$ through 
\[
\psi^{\bd}(\bfQ)
:=
E_{\bfQ}
\left[
m_1^{\bd}
\{d_1(H_1),H_1\}
\right],
\]
which serves as the observed-data target estimand for statistical estimation.
The EIC of $\psi^{\bd}(\bfQ)$ at $P_{\bfQ,\bfg}$ is
\begin{eqnarray*}
&&D^{\bd}(P_{\bfQ,\bfg})(O) \\
&=&
m_1^{\bd}\{d_1(H_1),H_1\}
-
\psi^{\bd}(\bfQ) \\ 
&&+
\sum_{k=1}^K
\frac{
\prod_{s=1}^k \I\{A_s=d_s(H_s)\}
}{
\prod_{s=1}^k g_s\{d_s(H_s)\mid H_s\}
}
\left[
m_{k+1}^{\bd}
\{d_{k+1}(H_{k+1}),H_{k+1}\} 
-
m_k^{\bd}(A_k,H_k) 
\right].
\end{eqnarray*}

\subsection{ADL-LTMLE Procedure}
\label{subsec:adl-ltmle-procedure}
The ADL-LTMLE is constructed recursively as follows.

First, we initialize the final pseudo-outcome by
$
m_{n,K+1}^{\bd,*}
\{d_{K+1}(H_{K+1}),H_{K+1}\} 
\equiv
Y$.
Then, for each $k=K,\ldots,1$, we proceed as follows.

At stage $k$, we use the targeted regression $m_{n,k+1}^{\bd,*}$ obtained at stage \(k+1\) to construct an estimate
$m_{n,k}^{\bd}(a_k,h_k)$ 
by regressing  
$
m_{n,k+1}^{\bd,*}
\{d_{k+1}(H_{k+1}),H_{k+1}\}$ on $A_k$ and $H_k$.

We then perform a targeting step at stage $k$. Using the
average-reference cumulative treatment probability
$\bar G_{k,n}^{\bd}$ defined above, let
$
H_{k,n}^{\bd}(O)
=
\frac{
\prod_{s=1}^k \I\{A_s=d_s(H_s)\}
}{
\bar G_{k,n}^{\bd}(H_k)
}$.
We fluctuate \(m_{n,k}^{\bd}\) along a logistic model:
\begin{align}
\operatorname{logit}
m_{n,k,\epsilon}^{\bd}(A_k,H_k)
=
\operatorname{logit}
m_{n,k}^{\bd}(A_k,H_k)
+
\epsilon H_{k,n}^{\bd}(O),    
\end{align}
where we use
$
m_{n,k+1}^{\bd,*}
\{d_{k+1}(H_{k+1}),H_{k+1}\}$
as the outcome. Let \(\hat\epsilon_{k,n}\) denote the estimated
fluctuation parameter and set
$
m_{n,k}^{\bd,*}
=
m_{n,k,\hat\epsilon_{k,n}}^{\bd}$. The targeted prediction
$m_{n,k}^{\bd,*}\{d_k(H_k),H_k\}$
then becomes the pseudo-outcome for estimating the regression at stage
\(k-1\). 

After completing the recursion, let
$
\mathbf m_n^{\bd,*}
=
\left(
m_{n,1}^{\bd,*},\ldots,m_{n,K}^{\bd,*}
\right)$
denote the resulting collection of targeted sequential regressions, and let
\(Q_{1,n}^*\) denote the empirical distribution of \(L_1\). We use
$
\bfQ_n^{*}
=
\left(
Q_{1,n}^*,
Q_{2,n}^{*},
\ldots,
Q_{K+1,n}^{*}
\right)
$
to denote the longitudinal distribution that induces the sequential regressions
\(\mathbf m_n^{\bd,*}\), which solves $\frac{1}{n}\sum_{i=1}^n D^{\bd}(P_{\bfQ_n^*,\bfg_n})(O_i)=0$.
The resulting ADL-LTMLE is
\[
\psi^{\bd}(\bfQ_n^*)
:=
\frac{1}{n}\sum_{i=1}^n
m_{n,1}^{\bd,*}
\{d_1(H_{1,i}),H_{1,i}\}.
\]


\subsection{Asymptotic Normality}
We now state the asymptotic normality of the ADL-LTMLE.
For any
integrable function \(f\), we use
$P_{\bfQ_0,\bfg_i}f = E\{f(O_i)\mid\mathcal F_{i-1}\}$
to denote the conditional expectation under the participant-specific
randomization mechanism \(\bfg_i\), given the previously observed data. Define an i.i.d. reference law by the average of
the participant-specific observed-data laws across participants in the adaptive experiment, i.e.,
$
\bar P_{\bfQ,\bfg_{1:n}}
:=
\frac{1}{n}\sum_{i=1}^n
P_{\bfQ,\bfg_i}$, which serves as
the reference law for characterizing the
semiparametric efficiency bound for regular estimators under adaptive experiments and for constructing the ADL-LTMLE. Moreover, we write
$
\bar P_{\bfQ,\bfg_{1:n}}
=
P_{\bfQ,\bar{\bfg}_n}$
as shorthand, where
$\bar{\bfg}_n=(\bar g_{1,n},\ldots,\bar g_{K,n})$
denotes the longitudinal treatment mechanism induced by the
average-reference law. In particular,   
$
\bar g_{k,n}(a_k\mid h_k)
:=
\frac{
\bar G_{k,n}(\bar a_k,\bar h_k)
}{
\bar G_{k-1,n}(\bar a_{k-1},\bar h_{k-1})
}
$, where $
\bar G_{k,n}(\bar a_k,\bar h_k)
:=
\frac{1}{n}
\sum_{i=1}^n
\prod_{s=1}^k
g_{s,i}(a_s\mid h_s)$, and $\bar G_{0,n}:=1$. Further details are provided in
Section~S1 of the Supplementary Material.

Let \(\bfQ_\infty\) and \(\bar{\bfg}_\infty\) denote fixed limits of \(\bfQ_n^*\) and \(\bar{\bfg}_n\), respectively, and define
$
G_k^{\bar{\bfg}_\infty,\bd}(h_k)
:=
\prod_{s=1}^k
\bar g_{s,\infty}\{d_s(h_s)\mid h_s\}$. 
Beyond the positivity and variance stabilization assumptions below, we additionally impose entropy and convergence conditions governing the
relevant function classes and the asymptotic behavior of $\bfQ_n^*$ and
$\bar{\bfg}_n$; their formal statements are provided in
Assumptions~
\ref{assumption:longitudinal-equicontinuity-conditions}--%
\ref{assumption:longitudinal-equicontinuity-convergence-g}
of the Supplementary Material.

\begin{assumption}[Strong Positivity of Average Cumulative Treatment Probabilities]
\label{assumption:longitudinal-average-reference-positivity}
There exists \(\zeta>0\) such that, for every \(n\),
$
\min_{k\in[K]}
\inf_{h_k\in\mathcal H_k}
\bar G_{k,n}^{\bd}(h_k)
\geq
\zeta$,
and 
$\min_{k\in[K]}
\inf_{h_k\in\mathcal H_k}
G_k^{\bar{\bfg}_\infty,\bd}(h_k)
\geq
\zeta$.

\end{assumption}

\begin{assumption}[Variance Stabilization]
\label{assumption:longitudinal-stabilized-variance-d}
\looseness=-1
There exists a constant  
$\sigma_{\bd}^2 > 0$
such that \\
$
\frac{1}{n}\sum_{i=1}^n
\operatorname{Var}_{P_{\bfQ_0,\bfg_i}}
\left[
D^{\bd}
(P_{\bfQ_\infty,\bar{\bfg}_{\infty}})(O)
\right]
\pto
\sigma_{\bd}^2$.
\end{assumption}

\begin{theorem}[Asymptotic Normality of the ADL-LTMLE]
\label{theorem:longitudinal-asymptotic-normality}

Suppose
Assumptions~\ref{assumption:longitudinal-average-reference-positivity},
\ref{assumption:longitudinal-stabilized-variance-d},
\ref{assumption:longitudinal-equicontinuity-conditions},
\ref{assumption:longitudinal-equicontinuity-convergence-Q}, and
\ref{assumption:longitudinal-equicontinuity-convergence-g}
hold. Then
\[
\sqrt n
\left\{
\psi^{\bd}(\bfQ_n^*)
-
\psi^{\bd}(\bfQ_0)
\right\}
\dto
N(0,\sigma_{\bd}^2).
\]
\end{theorem}

Let $\hat{\sigma}_{\bd,n}^2$ be a consistent estimate of $\sigma_{\bd}^2$. The standard error of $\psi^{\bd}(\bfQ_n^*)$ is given by $\hat\sigma_{\bd,n}/\sqrt{n}$. The $(1-\alpha)$ Wald confidence interval is $\psi^{\bd}(\bfQ_n^*) \pm z_{1-\alpha/2}\hat{\sigma}_{\bd,n}/\sqrt{n}$, where $z_{1-\alpha/2}$ is the $(1-\alpha/2)$-quantile of the standard
normal distribution.

\subsection{Semiparametric Efficiency}
\label{subsec:longitudinal-efficiency}

We further establish the semiparametric efficiency of the ADL-LTMLE following the efficiency result of
\citet{zhang2025efficient}, which relates the semiparametric efficiency notion for the adaptive
experiment $P_{\bfQ,\bfg_{1:n}}$ to that of a limiting i.i.d.\ reference
experiment with law $P_{\bfQ,\bar{\bfg}_\infty}$.

For any
$\bfQ=(Q_1,\ldots,Q_{K+1}) \in \mathcal{Q}$, define the tangent space
\[
\mathcal H(\bfQ,\bar{\bfg}_\infty)
=
\left\{
h=\sum_{k=1}^{K+1}h_k
\in L^2(P_{\bfQ,\bar{\bfg}_\infty}):
\begin{array}{l}
E_{\bfQ}\{h_1(L_1)\}=0,\\
E_{\bfQ}\{h_{k+1}(L_{k+1},A_k,H_k)
\mid A_k,H_k\}=0,\quad k\in[K]
\end{array}
\right\},
\]
equipped with the inner product
$
\langle h,\widetilde h\rangle_{\mathcal H(\bfQ,\bar{\bfg}_\infty)}
=
P_{\bfQ,\bar{\bfg}_\infty}(h\widetilde h)$.
\begin{assumption}[Reference-Law Norm Stability]
\label{assumption:longitudinal-hilbert-norm}
For every
\(h\in\mathcal H(\bfQ_0,\bar{\bfg}_\infty)\),
$P_{\bfQ_0,\bar{\bfg}_n}h^2
\pto 
P_{\bfQ_0,\bar{\bfg}_\infty}h^2$.
\end{assumption}
\begin{theorem}[Semiparametric Efficiency of the ADL-LTMLE]
\label{theorem:longitudinal-efficiency}

Suppose Assumptions \ref{assumption:longitudinal-average-reference-positivity},
\ref{assumption:longitudinal-stabilized-variance-d},
\ref{assumption:longitudinal-hilbert-norm},
\ref{assumption:longitudinal-equicontinuity-conditions},
\ref{assumption:longitudinal-equicontinuity-convergence-Q}, and 
\ref{assumption:longitudinal-equicontinuity-convergence-g} hold, and 
\(\bfQ_\infty=\bfQ_0\). Then the ADL-LTMLE is regular and efficient, 
\[
\sqrt n
\left\{
\psi^{\bd}(\bfQ_n^*)
-
\psi^{\bd}(\bfQ_0)
\right\}
\dto 
N\left(
0,
P_{\bfQ_0,\bar{\bfg}_\infty}
\left[
D^{\bd}(P_{\bfQ_0,\bar{\bfg}_\infty})^2
\right]
\right),
\]
i.e., the asymptotic variance of $\sqrt{n}\{ \psi^{\bd}(\bfQ_n^*) - \psi^{\bd}(\bfQ_0)\}$ equals the minimum attainable asymptotic variance among all regular estimators.
\end{theorem}

\section{Application to Multiple Time-to-Treatment-Initiation Effects}
\label{sec:example_estimand}
This section illustrates the general framework developed in
Sections~\ref{sec:general_framework} using a
class of time-to-treatment-initiation effects as a motivating example.
These effects compare initiating treatment at a given stage with delaying
initiation until the subsequent stage. Such comparisons reflect clinically relevant decision problems in many settings, such as cancer and HIV treatments 
(see, e.g., \cite{cone2020assessment,zhao2022antiretroviral,zhang2025optimal}).

\paragraph{Time-to-Treatment-Initiation Effects.}

Consider a class of treatment rules that vary in the timing of treatment initiation:
\[
\bd^\tau
=
(\underbrace{0,\ldots,0}_{\tau-1},
 \underbrace{1,\ldots,1}_{K-\tau+1}),
\qquad
\bd^{K+1}=(0,\ldots,0),
\]
and consider the $\tau$-specific causal estimand
$
\Psi^\tau(\bfQ)
=
E(Y^{\bd^\tau})-E(Y^{\bd^{\tau+1}}),
$
which compares initiating treatment at stage $\tau$ with delaying
initiation until stage $\tau+1$. For $\tau = K$, this compares initiating treatment at the final stage $K$ with never initiating treatment. Under the identification conditions in Assumption \ref{assumption:causal-identification-d},
the target estimands of interest are
\[
\psi^\tau(\bfQ)
=
E\left[
m_1^{\bd^\tau}\{d_1^\tau(H_1),H_1\}
-
m_1^{\bd^{\tau+1}}
\{d_1^{\tau+1}(H_1),H_1\}
\right], \quad \tau \in [K].
\]

\paragraph{Longitudinal Adaptive Design.}
For a treatment rule $\bd$, define the stage-$k$ sequential-regression
residual
$
\Delta_k^{\bd}(O)
=
m_{k+1}^{\bd}
\{d_{k+1}(H_{k+1}),H_{k+1}\}
-
m_k^{\bd}(A_k,H_k)$,
where $m_{k+1}^{\bd}
\{d_{k+1}(H_{k+1}),H_{k+1}\} \equiv Y$ when $k=K$. The EIC of
$\psi^\tau$ admits the decomposition
$
D^\tau(P_{\bfQ,\bfg})
=
D_0^\tau(P_{\bfQ,\bfg})
+
\sum_{k=1}^K
D_k^\tau(P_{\bfQ,\bfg}),
$
where
$
D_0^\tau
=
m_1^{\bd^\tau}\{d_1^\tau(H_1),H_1\}
-
m_1^{\bd^{\tau+1}}
\{d_1^{\tau+1}(H_1),H_1\}
-
\psi^\tau,
$
and
$
D_k^\tau(P_{\bfQ,\bfg})(O)
=
\frac{
\Phi_k^\tau(\bfQ)(O)
}{
\prod_{s=1}^k g_s(A_s\mid H_s)
}$,
with
$
\Phi_k^\tau(\bfQ)(O)
=
\I(\bar A_k=\bd_{1:k}^{\tau})
\Delta_k^{\bd^\tau}(O)
-
\I(\bar A_k=\bd_{1:k}^{\tau+1})
\Delta_k^{\bd^{\tau+1}}(O).
$

This illustrates that this class of estimands admits the general stage-wise
representation required for applying the design optimization framework in Section~\ref{subsec:general_recursive}. Specifically, for each
stage $k$ and action $a$, let $\{\sigma_{a,k}^\tau(h_k)\}^2$ denote
the conditional variance of the downstream EIC components from $D_k^{\tau}$ to $D_K^{\tau}$
after separating out $g_k(a\mid h_k)^{-1}$. Conditional on
$H_k=h_k$, the component of the efficiency bound depending on $g_k$ is
$
\frac{\{\sigma_{1,k}^\tau(h_k)\}^2}{g_k(1\mid h_k)}
+
\frac{\{\sigma_{0,k}^\tau(h_k)\}^2}{g_k(0\mid h_k)}$.
Using the backward-recursive strategy to optimize the longitudinal design, the oracle randomization scheme at stage $k=\tau$ is 
\begin{equation}
g_k^*(1\mid h_k)
=
\frac{\sigma_{1,k}^{\tau}(h_k)}
     {\sigma_{0,k}^{\tau}(h_k) + \sigma_{1,k}^{\tau}(h_k)},
\quad
g_k^*(0\mid h_k)
=
\frac{\sigma_{0,k}^{\tau}(h_k)}
     {\sigma_{0,k}^{\tau}(h_k) + \sigma_{1,k}^{\tau}(h_k)}.
\label{eq:single_contrast_neyman}
\end{equation}
By contrast, for stages $k \neq \tau$, the oracle design becomes degenerate: 
\[
g_k^*(1\mid H_k)\equiv 0 \quad\text{for } k<\tau,
\quad
g_k^*(1\mid H_k)\equiv 1 \quad\text{for } k>\tau.
\]
\looseness=-1
Intuitively, the oracle design keeps all participants on the treatment history
shared by the two regimes before stage $\tau$, randomizes at the stage
where the regimes diverge according to the variance-minimizing allocation
determined by the relative action-specific conditional variability, and assigns the common treatment prescribed by both regimes for later stages. As a result, it maximally preserves the longitudinal support needed to
learn both treatment rules while randomizing treatment at the stage where
they diverge according to the optimal allocation defined by the efficiency
criterion for estimating their contrast.

\looseness=-1
However, when the goal is to learn multiple target estimands concurrently, such as the
collection of treatment-initiation effects
$\{\psi^\tau:\tau\in[K]\}$, an oracle design optimized for one estimand
may substantially reduce the information available for others. These trade-offs are particularly substantial in longitudinal
settings, where treatment is assigned sequentially and each stage-specific randomization decision directly shapes the data trajectories at subsequent stages. Consequently, favoring the optimal design for one estimand can irreversibly restrict the trajectories needed to estimate estimands that rely on sufficient data support at other stages, leading to limited support, inflated variance, or even loss of identification. This tension is explicit in the treatment-initiation-effect example, where estimand-specific oracle designs concentrate randomization at different stages and thereby impose competing demands on longitudinal data support. The joint optimization and the corresponding adaptive design framework in Section~\ref{subsec:joint_optimization} therefore provide a principled approach to
navigating these trade-offs.

\section{Simulation Studies}
\label{sec:simulation}
We conduct simulation studies to evaluate the proposed longitudinal designs and estimators in the illustrative example of learning multiple treatment-initiation effects
$(\psi^1,\ldots,\psi^K)$. Specifically, we assess the efficiency gains of the proposed adaptive designs developed in Section~\ref{sec:general_framework} relative
to fixed non-adaptive randomization and the oracle design, as well as the finite-sample performance of the proposed estimators.

\paragraph{Setup.}
\looseness=-1
In the primary simulations, $n=2000$ participants are enrolled
sequentially in $T=5$ equally sized cohorts of 400 participants. We consider
longitudinal settings with either two or three treatment stages
($K=2$ and $K=3$), with discrete time-varying covariates and a continuous outcome. The complete data-generating
mechanisms are provided in Section~\ref{supp_sec:dgd} of the Supplementary
Material.

For each setting, we consider three user-specified weight vectors
$w=(w^1,\ldots,w^K)$ in the joint efficiency criterion. For $K=2$, we use
$(1/2,1/2)$, $(1/3,2/3)$, and $(2/3,1/3)$; for $K=3$, we use
$
(1/3,1/3,1/3)$, $(1/6,1/3,1/2)$ and $(1/2,1/3,1/6)$.

We compare the following three designs:
(a) a fixed non-adaptive design that assigns treatment with probability
$0.5$ at every stage;
(b) an oracle design defined by recursively computing the optimal stage-specific randomization
probabilities using the true data-generating mechanism; and
(c) the proposed adaptive design, which starts from equal randomization and then estimates and tracks the oracle design using the accrued data.

\paragraph{Estimation.}
\looseness=-1
For each Monte Carlo replicate and design, we estimate the collection of
treatment-initiation effects
$\{\psi^\tau:\tau\in[K]\}$ using the ADL-LTMLE developed in
Section~\ref{section:adl-ltmle}, with the outcome rescaled to $[0,1]$ using the observed minimum and maximum
and the resulting estimates transformed back to the original scale. In the primary simulations, the initial
sequential outcome regressions are estimated using empirical means within
strata defined by the observed covariate and treatment histories available
before each treatment stage.
As a comparator, we also implement an adaptive-design LTMLE
(AD-LTMLE) based on the participant-specific randomization mechanisms
$\bfg_i$, which solves
$
\frac{1}{n}\sum_{i=1}^n
D^\tau\!\left(P_{\bfQ_n^*,\bfg_i}\right)(O_i)=0$.
This estimator is asymptotically normal under standard conditions for
adaptive experiments, but it may exhibit greater finite-sample variability
and require stronger design-stability conditions than the ADL-LTMLE based on the
design $\bar{\bfg}_n$ induced by the
average-reference law $\bar P_{\bfQ,\bfg_{1:n}}$  \citep{zhang2025efficient}.
The variance estimators for the ADL-LTMLE and AD-LTMLE are 
$
\frac{1}{n}
\sum_{t=1}^T
\frac{n_t}{n}\,
\widehat{\operatorname{Var}}_{i\in\mathcal I_t}
D^\tau\!\left(P_{\bfQ_n^*,\bar{\bfg}_n}\right)
$
and
$
\frac{1}{n}
\sum_{t=1}^T
\frac{n_t}{n}\,
\widehat{\operatorname{Var}}_{i\in\mathcal I_t}
D^\tau\!\left(P_{\bfQ_n^*,\bfg_i}\right),
$
respectively, where $\mathcal I_t$ denotes the set of participant indices in cohort $t$.
Under the non-adaptive and oracle designs, where the randomization
strategy is fixed, the two estimators coincide with the standard
i.i.d.\ LTMLE \citep{van2011LTMLE}.
\looseness = -1
We also evaluate two one-step estimators constructed using $\bar{\bfg}_n$ and $\bfg_i$, respectively, as analogues of ADL-LTMLE and
AD-LTMLE, with more details provided in the
Supplementary Material.

\paragraph{Results.}
Tables~\ref{tab:K2_bar_g_ltmle} and~\ref{tab:K3_bar_g_ltmle}
summarize the results over 500 Monte Carlo replicates for $K=2$ and
$K=3$, respectively. The columns Var$_j$ report the empirical variances
of the estimators of $\psi^j$, whereas TarVar reports the corresponding
weighted average of these variances using the weight vector specified in
the same row. The columns Cov$_j$ report the empirical coverage of $95\%$ Wald confidence intervals, whereas OCov$_j$ report the coverage obtained by using the Monte Carlo standard deviation
of the point estimates to construct the confidence intervals.

For $K=2$, the proposed adaptive design yields substantial efficiency
gains relative to the fixed design under both estimators. For
$w=(1/2,1/2)$, the values of TarVar under the fixed design, adaptive design with AD-LTMLE, adaptive design with 
ADL-LTMLE, and the oracle design are $0.029$, $0.022$,
$0.019$, and $0.017$, respectively. For the other weights, the corresponding values are
$0.026$, $0.019$, $0.017$, and $0.015$ for $w=(1/3,2/3)$, and
$0.033$, $0.025$, $0.021$, and $0.019$ for $w=(2/3,1/3)$.
Thus, even when paired with the AD-LTMLE, the adaptive design reduces
TarVar by approximately $24\%$--$26\%$ relative to the fixed design. Under the same adaptive design,
the ADL-LTMLE further reduces TarVar, resulting in overall reductions of approximately
$34\%$--$36\%$ relative to the fixed design and bringing performance
closer to that of the oracle design, which  yields variance reduction by approximately $41\%$--$42\%$.

The efficiency gains of the proposed designs are more substantial for $K=3$. For
$w=(1/3,1/3,1/3)$, the TarVar values under the fixed design, adaptive design with AD-LTMLE,
adaptive design with ADL-LTMLE, and oracle design are $0.056$, $0.031$,
$0.026$, and $0.022$, respectively. The corresponding values are
$0.062$, $0.032$, $0.024$, and $0.021$ for
$w=(1/6,1/3,1/2)$, and $0.050$, $0.029$, $0.024$, and $0.021$
for $w=(1/2,1/3,1/6)$. Pairing the proposed adaptive design with the
AD-LTMLE therefore reduces TarVar by approximately $42\%$--$48\%$
relative to the fixed design. The ADL-LTMLE provides a further variance reduction under the same adaptive design, yielding
overall reductions of approximately $52\%$--$61\%$. The resulting performance is close to the corresponding oracle benchmark,
which reduces TarVar by approximately $58\%$--$66\%$. 

The results also show that changing the weight vector shifts estimation efficiency toward the estimands assigned larger weights. For
$K=2$, under the oracle design, increasing
the weight on $\psi^2$ from $w=(2/3,1/3)$ to $w=(1/3,2/3)$ reduces
Var$_2$ from $0.012$ to $0.010$, while Var$_1$ increases slightly from $0.023$ to $0.024$.
The trade-offs are more pronounced for $K=3$. Shifting from equal weights to $w=(1/6,1/3,1/2)$ prioritizes
$\psi^3$ and reduces Var$_3$ from $0.020$ to $0.017$ in the oracle design,
while Var$_1$ increases from $0.020$ to $0.026$. Conversely, shifting from equal weights to $w=(1/2,1/3,1/6)$ prioritizes
$\psi^1$ and reduces Var$_1$ from $0.020$ to
$0.016$, while Var$_3$ increases from $0.020$ to
$0.025$. The adaptive designs follow the same qualitative pattern and achieve performance close to that of the oracle design.

Supplementary Tables S3--S6 summarize
sensitivity analyses with respect to the sample size and the model specification. For both $K=2$ and $K=3$, TarVar decreases substantially as the total sample size increases from 1,000 to 2,000 and then to 4,000. While misspecification of the outcome regressions increases variance, coverage remains near the nominal level across model specifications and sample sizes, supporting the robustness of ADL-LTMLE to outcome-regression misspecification.

\input{table_K2_LTMLE}
\input{table_K3_LTMLE}
\section{Conclusions}

\label{sec:conclusion}
In this work, we develop a general framework for constructing longitudinal adaptive
experiments aimed at learning one or more statistical estimands with semiparametric efficient inference.
We characterize oracle longitudinal designs that minimize semiparametric efficiency
bounds using backward-recursive solutions of treatment randomization schemes across treatment stages, and propose an extended design optimization criterion that balances
estimation efficiency of multiple estimands and navigates their trade-offs.
This approach applies broadly to the identification of oracle designs and the construction of adaptive designs for a wide class of target estimands, including various timing-of-treatment-initiation effects, treatment-specific means at multiple time points, and mean outcomes under dynamic or stochastic treatment rules. 
We further establish an asymptotically normal, semiparametric efficient ADL-LTMLE framework for asymptotically
normal and semiparametric efficient estimation of statistical estimands from dependent data collected from adaptive experiments, without relying on parametric model assumptions. 
Simulation studies show that the
proposed longitudinal designs achieve substantial efficiency gains over non-adaptive
designs with the performance close to that of the unknown oracle design, especially when combined with the efficient ADL-LTMLE.  

Overall, this work contributes a general efficiency-oriented framework for the design and estimation of adaptive experiments in longitudinal settings involving multi-stage, time-varying treatments. Future work may explore applications of the proposed framework to other
statistical learning objectives and extend it to a broader range of experimental settings and application domains.

\clearpage
\section*{Supplementary Material}

\bibliographystyle{chicago}      
\setcounter{section}{0}
\setcounter{equation}{0}
\def\theequation{S\arabic{section}.\arabic{equation}}
\def\thesection{S\arabic{section}}

\fontsize{12}{14pt plus.8pt minus .6pt}\selectfont

\setcounter{table}{0}
\renewcommand{\thetable}{S\arabic{table}}

\setcounter{assumption}{0}
\renewcommand{\theassumption}{S\arabic{assumption}}

\section{Theoretical Analyses of ADL-LTMLE}
This section provides theoretical analyses of ADL-LTMLE. We first characterize the treatment randomization mechanisms induced the reference distribution used to construct the estimator, then establish its asymptotic normality and semiparametric efficiency.
\subsection{Stage-wise Treatment Distribution under the Reference Distribution}
\label{supp_sec:bar-g-derivation}
Recall that 
$
\bar P_{\bfQ,\bfg_{1:n}}
:=
\frac{1}{n}\sum_{i=1}^n P_{\bfQ,\bfg_i}$,
where
\(\bfg_i=(g_{1,i},\ldots,g_{K,i})\)
denotes the collection of
randomization mechanisms used to generate participant \(i\)'s treatment sequence. Here,
$\bfQ=(Q_1,Q_2,\ldots,Q_{K+1})$
collects the stage-wise nuisance components, with densities denoted by
\(q_1,\ldots,q_{K+1}\). The density of
\(\bar P_{\bfQ,\bfg_{1:n}}\) is
\[
\bar p_{\bfQ,\bfg_{1:n}}(o)
=
\frac{1}{n}\sum_{i=1}^n
\left[
q_1(l_1)
\prod_{k=1}^{K}
g_{k,i}(a_k\mid h_k)
q_{k+1}(l_{k+1}\mid a_k,h_k)
\right].
\]

For any fixed longitudinal treatment mechanism
$
\bfg=(g_1,\ldots,g_K),
$
define its cumulative treatment probability through stage \(k\) by
$
G_k^{\bfg}(\bar a_k\mid h_k)
:=
\prod_{s=1}^k
g_s(a_s\mid h_s).
$
For participant \(i\), write
\[
G_{k,i}(\bar a_k\mid h_k)
:=
G_k^{\bfg_i}(\bar a_k\mid h_k)
=
\prod_{s=1}^k
g_{s,i}(a_s\mid h_s).
\]

For any \(1\leq k\leq K\), the joint density of
\((\bar L_k,\bar A_k)\) under
\(\bar P_{\bfQ,\bfg_{1:n}}\), obtained by marginalizing over future variables,
is
\[
p_{\bar P}(\bar l_k,\bar a_k)
=
\left[
q_1(l_1)
\prod_{s=2}^{k}
q_s(l_s\mid a_{s-1},h_{s-1})
\right]
\left[
\frac{1}{n}\sum_{i=1}^n
G_{k,i}(\bar a_k\mid h_k)
\right].
\]
Similarly,
\[
p_{\bar P}(\bar l_k,\bar a_{k-1})
=
\left[
q_1(l_1)
\prod_{s=2}^{k}
q_s(l_s\mid a_{s-1},h_{s-1})
\right]
\left[
\frac{1}{n}\sum_{i=1}^n
G_{k-1,i}(\bar a_{k-1}\mid h_{k-1})
\right]
\]
for $2 \leq k \leq K$.
For $k=1$, $p_{\bar P}(\bar l_k,\bar a_{k-1})$ reduces to $p_{\bar P}(\bar l_k)$, which equals $q_1(l_1)$.

Define the cumulative treatment probability under the finite-sample
average-reference distribution by
$
\bar G_{k,n}(\bar a_k\mid h_k)
:=
\frac{1}{n}\sum_{i=1}^n
G_{k,i}(\bar a_k\mid h_k).
$
Therefore, the stage-wise treatment distribution induced by the
average-reference law for every $k \in [K]$ is
\begin{align}
\bar g_{k,n}(a_k\mid h_k)
=
\frac{
p_{\bar P}(\bar l_k,\bar a_k)
}{
p_{\bar P}(\bar l_k,\bar a_{k-1})
}
=
\frac{
\bar G_{k,n}(\bar a_k\mid h_k)
}{
\bar G_{k-1,n}(\bar a_{k-1}\mid h_{k-1})
}.
\end{align}
In particular,
$
\bar g_{1,n}(a_1\mid h_1)
=
\frac{1}{n}\sum_{i=1}^n
g_{1,i}(a_1\mid h_1)$.

Let
$
\bar{\bfg}_n
:=
(\bar g_{1,n},\ldots,\bar g_{K,n})$.
By construction,
\[
G_k^{\bar{\bfg}_n}(\bar a_k\mid h_k)
=
\bar G_{k,n}(\bar a_k\mid h_k)
=
\frac{1}{n}\sum_{i=1}^n
G_k^{\bfg_i}(\bar a_k\mid h_k).
\]
Thus, the cumulative treatment probability induced by the average-reference
distribution \(\bar P_{\bfQ,\bfg_{1:n}}\) is the average of the
participant-specific cumulative treatment probabilities. For notational
convenience in the following sections, we therefore write
$
P_{\bfQ,\bar{\bfg}_n}
:=
\bar P_{\bfQ,\bfg_{1:n}}$,
which infers a fixed reference experiment consists of independent
observations drawn from \(P_{\bfQ,\bar{\bfg}_n}\) where each participant has stage-wise treatments sequentially generated according to
the fixed mechanism \(\bar{\bfg}_n\). 

\subsection{Asymptotic Normality}
We use the example of the mean outcome under a longitudinal dynamic treatment rule
\(\bd=(d_1,\ldots,d_K)\) as the target estimand of interest, which falls within the broader ADL-TMLE framework of
\citet{zhang2025efficient} that considers general target parameters of the
form $\Psi(\bfQ)$ for data likelihoods factorizing as
$p(o_1,\ldots,o_n)=\prod_{i=1}^n \mathbf q(o_i) \bfg_i(o_i)$, where $\mathbf q$ is the density of some
$\bfQ \in \mathcal Q$ that defines the target estimand.

Let 
$\psi^{\bd}(P_{\bfQ,\bfg})=E_{P_{\bfQ,\bfg}}(Y^{\bd})$ denote the mean counterfactual outcome under \(\bd\).
Since the estimand is only defined on $P_{\bfQ,\bfg}$ through $\bfQ$, we also express the estimand by $\psi^{\bd}(\bfQ)$.

\begin{theorem}[Decomposition of the ADL-LTMLE]
\label{theorem:longitudinal-tmle-decomposition}
\[
\psi^{\bd}(\bfQ_n^*)
-
\psi^{\bd}(\bfQ_0)
=
M_{1,n}^{\bd}(\bfQ_\infty, \bar{\bfg}_\infty)
+
M_{2,n}^{\bd}(\bfQ_n^*,\bfQ_{\infty},\bar{\bfg}_n)
+
M_{3,n}^{\bd}(\bfQ_\infty,\bar{\bfg}_n,\bar{\bfg}_{\infty}),
\]
where
\[
M_{1,n}^{\bd}(\bfQ_\infty, \bar{\bfg}_\infty)
:=
\frac{1}{n}\sum_{i=1}^n
\left[
D^{\bd}(P_{\bfQ_\infty,\bar{\bfg}_\infty})(O_i)
-
P_{\bfQ_0,\bfg_i}
D^{\bd}(P_{\bfQ_\infty,\bar{\bfg}_\infty})
\right],
\]
\[
\begin{aligned}
M_{2,n}^{\bd}(\bfQ_n^*,\bfQ_{\infty},\bar{\bfg}_n)
&:=
\frac{1}{n}\sum_{i=1}^n
\Bigg\{
\left[
D^{\bd}(P_{\bfQ_n^*,\bar{\bfg}_n})(O_i)
-
P_{\bfQ_0,\bfg_i}
D^{\bd}(P_{\bfQ_n^*,\bar{\bfg}_n})
\right]
\\
&\hspace{5.2em}-
\left[
D^{\bd}(P_{\bfQ_\infty,\bar{\bfg}_n})(O_i)
-
P_{\bfQ_0,\bfg_i}
D^{\bd}(P_{\bfQ_\infty,\bar{\bfg}_n})
\right]
\Bigg\},
\end{aligned}
\]
and
\[
\begin{aligned}
M_{3,n}^{\bd}(\bfQ_\infty,\bar{\bfg}_n,\bar{\bfg}_{\infty})
&:=
\frac{1}{n}\sum_{i=1}^n
\Bigg\{
\left[
D^{\bd}(P_{\bfQ_\infty,\bar{\bfg}_n})(O_i)
-
P_{\bfQ_0,\bfg_i}
D^{\bd}(P_{\bfQ_\infty,\bar{\bfg}_n})
\right]
\\
&\hspace{5.2em}-
\left[
D^{\bd}(P_{\bfQ_\infty,\bar{\bfg}_\infty})(O_i)
-
P_{\bfQ_0,\bfg_i}
D^{\bd}(P_{\bfQ_\infty,\bar{\bfg}_\infty})
\right]
\Bigg\}.
\end{aligned}
\]
\end{theorem}

\begin{proof}
Define the sequential
regressions recursively by setting
\[
m_{K+1}^{\bd}\{d_{K+1}(H_{K+1}),H_{K+1}\}=Y,
\]
and, for \(k=1,\ldots,K\),
\[
m_k^{\bd}(a_k,h_k)
=
E_{\bfQ}
\left[
m_{k+1}^{\bd}\{d_{k+1}(H_{k+1}),H_{k+1}\}
\mid
A_k=a_k,\ H_k=h_k
\right].
\]

For any longitudinal treatment mechanism \(\bfg\), define
$
G_k^{\bfg,\bd}(H_k)
:=
\prod_{s=1}^k
g_s\{d_s(H_s)\mid H_s\}$.
For participant \(i\), we write
$
G_{k,i}^{\bd}(H_k)
:=
G_k^{\bfg_i,\bd}(H_k)$.
The corresponding cumulative probability under the finite-sample
average-reference mechanism is
$
\bar G_{k,n}^{\bd}(H_k)
:=
G_k^{\bar{\bfg}_n,\bd}(H_k)
=
\frac{1}{n}\sum_{i=1}^n
G_{k,i}^{\bd}(H_k)$.

The treatment-rule-specific EIC evaluated under any fixed $P_{\bfQ,\bfg}$ is
\[
\begin{aligned}
D^{\bd}(P_{\bfQ,\bfg})(O)
&=
m_1^{\bd}\{d_1(H_1),H_1\}
-
\psi^{\bd}(\bfQ)
\\
&\quad+
\sum_{k=1}^{K}
\frac{
\I\{\bar A_k=\bar d_k\}
}{
\bar G_{k,n}^{\bd}(H_k)
}
\left[
m_{k+1}^{\bd}\{d_{k+1}(H_{k+1}),H_{k+1}\}
-
m_k^{\bd}(A_k,H_k)
\right].
\end{aligned}
\]
Under the event \(\I\{\bar A_k=\bar d_k\}=1\), the term
\(m_k^{\bd}(A_k,H_k)\) equals \(m_k^{\bd}\{d_k(H_k),H_k\}\).

For participant \(i\), define the cumulative probability of following
\(\bd\) through stage \(k\) under that participant's randomization mechanism by
$
G_{k,i}^{\bd}(H_k)
=
\prod_{s=1}^k
g_{s,i}\{d_s(H_s)\mid H_s\}$.
We also define the stage-\(k\) residual under \(\bfQ\) by
$
\Delta_k^{\bd}(\bfQ)
:=
m_{k+1}^{\bd}\{d_{k+1}(H_{k+1}),H_{k+1}\}
-
m_k^{\bd}\{d_k(H_k),H_k\}.
$

For each \(i\), applying von Mises expansion with respect to $P_{\bfQ,\bar{\bfg}_n}$ and $P_{\bfQ_0,\bfg_i}$ yields
\[
\begin{aligned}
\psi^{\bd}(P_{\bfQ,\bar{\bfg}_n})-\psi^{\bd}(P_{\bfQ_0,\bfg_i})
&=
\int
D^{\bd}(P_{\bfQ,\bar{\bfg}_n})(o)
\,d\{P_{\bfQ,\bar{\bfg}_n}-P_{\bfQ_0,\bfg_i}\}(o)
\\
&\quad+
R_i^{\bd}(\bfQ,\bfQ_0;\bar{\bfg}_n,\bfg_i).
\end{aligned}
\]

For the mean outcome under a longitudinal dynamic treatment rule $\bd$, the remainder takes the form
\[
R_i^{\bd}(\bfQ,\bfQ_0;\bar{\bfg}_n,\bfg_i)
=
\sum_{k=1}^{K}
E_{\bfQ_0}^{\bd}
\left[
\frac{
G_{k,i}^{\bd}(H_k) - \bar G_{k,n}^{\bd}(H_k)
}{
\bar G_{k,n}^{\bd}(H_k)
}
\Delta_k^{\bd}(\bfQ)
\right],
\]
where
$
G_{k,i}^{\bd}(H_k)
=
\prod_{s=1}^{k}
g_{s,i}\{d_s(H_s)\mid H_s\},
$
$
\bar G_{k,n}^{\bd}(H_k)
=
\frac{1}{n}\sum_{i=1}^{n}G_{k,i}^{\bd}(H_k),
$
and
$
\Delta_k^{\bd}(\bfQ)
=
m_{k+1}^{\bd}\{d_{k+1}(H_{k+1}),H_{k+1}\}
-
m_k^{\bd}\{d_k(H_k),H_k\}.
$

Here, \(E_{\bfQ_0}^{\bd}\) denotes expectation under the distribution 
that sets treatment according to \(\bd\) while keeping the longitudinal
covariate and outcome law at \(\bfQ_0\). 
By averaging over $i = 1,\ldots,n$, we obtain
\[
\begin{aligned}
\psi^{\bd}(\bfQ)-\psi^{\bd}(\bfQ_0)
&=
\frac{1}{n}\sum_{i=1}^{n}
\int
D^{\bd}(P_{\bfQ,\bar{\bfg}_n})(o)
\,d\{P_{\bfQ,\bar{\bfg}_n}-P_{\bfQ_0,\bfg_i}\}(o)
\\
&\quad+
\frac{1}{n}\sum_{i=1}^{n}
\sum_{k=1}^{K}
E_{\bfQ_0}^{\bd}
\left[
\frac{
G_{k,i}^{\bd}(H_k) - \bar G_{k,n}^{\bd}(H_k)
}{
\bar G_{k,n}^{\bd}(H_k)
}
\Delta_k^{\bd}(\bfQ)
\right] \\
&=
\frac{1}{n}\sum_{i=1}^{n}
\int
D^{\bd}(P_{\bfQ,\bar{\bfg}_n})(o)
\,d\{P_{\bfQ,\bar{\bfg}_n}-P_{\bfQ_0,\bfg_i}\}(o)
\\
&= - P_{\bfQ_0,\bar{\bfg}_n}
D^{\bd}(P_{\bfQ,\bar{\bfg}_n}),
\end{aligned}
\]
where the second equality holds since 
$\bar G_{k,n}^{\bd}(H_k)
=
\frac{1}{n}\sum_{i=1}^{n}G_{k,i}^{\bd}(H_k)$, and the last equality uses the fact that 
$
P_{\bfQ,\bfg}
D^{\bd}(P_{\bfQ,\bfg})=0$ for any $P_{\bfQ,\bfg}$.

Since an ADL-LTMLE for estimating $\psi^{\bd}(\bfQ_0)$ updates an initial estimate $\bfQ_n$ to obtain a \(\bfQ_n^*\) that solves
$
\frac{1}{n} \sum_{i=1}^n D^{\bd}(P_{\bfQ_n^*,\bar{\bfg}_n})(O_i)
= 0$, we
have
\begin{equation*}
    \psi^{\bd}(\bfQ_n^*)-\psi^{\bd}(\bfQ_0)
= \frac{1}{n} \sum_{i=1}^n D^{\bd}(P_{\bfQ_n^*,\bar{\bfg}_n}) - P_{\bfQ_0,\bar{\bfg}_n}
D^{\bd}(P_{\bfQ_n^*,\bar{\bfg}_n}).
\end{equation*}

The final decomposition follows from some algebra that adds and subtracts other corresponding terms. 
\end{proof}

\begin{theorem}[Asymptotic Normality of $M_{1,n}^{\bd}$]
\label{theorem:longitudinal-asymptotic-normality-M1}
Suppose Assumptions~\ref{assumption:longitudinal-average-reference-positivity} and \ref{assumption:longitudinal-stabilized-variance-d} hold, then 
$
\sqrt{n} M_{1,n}^{\bd}(\bfQ_\infty, \bar{\bfg}_\infty) \dto N(0, \sigma_{\bd}^2).
$
\end{theorem}

\begin{proof}
Note that 
\(M_{1,n}^{\bd}(\bfQ_\infty, \bar{\bfg}_\infty)\) is an average of martingale differences sequences. 
Under Assumption~\ref{assumption:longitudinal-average-reference-positivity},
$D^{\bd}(P_{\bfQ_\infty,\bar{\bfg}_\infty})$
is uniformly bounded, therefore satisfying the conditional Lindeberg condition, i.e., for every $\epsilon >0$, $
\frac{1}{n}\sum_{i=1}^n
E\left[
Z_{n,i}^2
\I\left\{
|Z_{n,i}|>\epsilon\sqrt n
\right\}
\mid
\bO(i-1)
\right]
\pto  0,
$
where
$
Z_{n,i}
=
D(P_{\bfQ_\infty,\bar{\bfg}_\infty})(O_i)
-
P_{\bfQ_0,\bfg_i}
D(P_{\bfQ_\infty,\bar{\bfg}_\infty})$. 

Combined with the variance stabilization assumption stated in Assumption \ref{assumption:longitudinal-stabilized-variance-d}, one can directly apply Martingale Central Limit Theorem (see, e.g., \cite{hall2014martingale}) and obtain the asymptotic normality. 
\end{proof}

Now we establish equicontinuity conditions for \(M_{2,n}^{\bd}\) and \(M_{3,n}^{\bd}\). 

For \(\bfQ=(Q_1,\ldots,Q_{K+1})\in\cQ\), define the collection of sequential regression
functions associated with rule \(\bd\) by
$
\mathbf m^{\bd}(\bfQ)
=
\left(
m_1^{\bd}(\bfQ),\ldots,m_K^{\bd}(\bfQ)
\right)$,
and let
$
\mathcal M^{\bd}
=
\left\{
\mathbf m^{\bd}(\bfQ):
\bfQ\in\cQ
\right\}$ equipped with the supremum norm defined by 
$
\left\|
\mathbf m^{\bd}
\right\|_\infty
=
\max_{k\in[K]}
\left\|
m_k^{\bd}
\right\|_\infty$.
Let \(\cG\) denote the fixed class of treatment mechanisms satisfying Assumption \ref{assumption:longitudinal-average-reference-positivity}, equipped
with
$
\left\|
\bfg
\right\|_\infty
=
\max_{k\in[K]}
\sup_{a_k,h_k}
\left|
g_k(a_k\mid h_k)
\right|$.
For a class of functions \(\mathcal H\), let
\(N_\infty(u,\mathcal H)\) denote its \(u\)-covering number under the
corresponding supremum norm, and define the entropy integral 
$
J_\infty(\epsilon,\mathcal H)
:=
\int_0^\epsilon
\sqrt{
\log\left[
1+N_\infty(u,\mathcal H)
\right]
}
\,du$.

\begin{assumption}[Entropy Conditions]
\label{assumption:longitudinal-equicontinuity-conditions}

$J_\infty(\epsilon,\mathcal M^{\bd})$ and 
$J_\infty(\epsilon,\cG)$ are finite and converge to zero as $\epsilon \to 0$.
\end{assumption}
\begin{assumption}[Convergence of $\bfQ_n^*$]
\label{assumption:longitudinal-equicontinuity-convergence-Q}
Define $\rho(x)=e^x-x-1$ and $\widetilde D^{\bd}(P_{\bfQ,\bar{\bfg}_n})(O_i)
= m_1^{\bd}(\bfQ)(O_i)+\sum_{k=1}^K D_k^{\bd}(P_{\bfQ,\bar{\bfg}_n})(O_i)$ for any $\bfQ \in \mathcal Q$.
Let  
\[
\begin{aligned}
s_{2,n}^{\bd}
:=
\Bigg\{
\frac{1}{n}\sum_{i=1}^n
E\Bigg[
\rho\Bigg(
\frac{
\big|D^{\bd}(P_{\bfQ^*_n,\bar{\bfg}_n})(O_i)-D^{\bd}(P_{\bfQ_\infty,\bar{\bfg}_n})(O_i)\big|
}{
c_2
}
\Bigg)
\Bigm|
\mathcal F_{i-1}
\Bigg]
\Bigg\}^{1/2},
\end{aligned}
\]
where $c_2 = 8K/\zeta+8K^2/\zeta^2$.
Assume that 
$s_{2,n}^{\bd}
=
O_P(\delta_{2,n})$ for some \(\delta_{2,n}\to0\).
\end{assumption}

\begin{assumption}[Convergence of $\bar{\bfg}_n$]
\label{assumption:longitudinal-equicontinuity-convergence-g}
Define
\[
\begin{aligned}
s_{3,n}^{\bd}
:=
\Bigg\{
\frac{1}{n}\sum_{i=1}^n
E\Bigg[
\rho\Bigg(
\frac{
\left|
D^{\bd}(P_{\bfQ_\infty,\bar{\bfg}_n})(O_i)
-
D^{\bd}(P_{\bfQ_\infty,\bar{\bfg}_\infty})(O_i)
\right|
}{
c_3
}
\Bigg)
\Bigm|
\mathcal F_{i-1}
\Bigg]
\Bigg\}^{1/2},
\end{aligned}
\]
where $c_3 = 4K(K+1)/\zeta^2$.
Assume that
$
s_{3,n}^{\bd}
=
O_P(\delta_{3,n})$ for some 
\(\delta_{3,n}\to0\).
\end{assumption}

We also refer to the following lemma of a general martingale equicontinuity result of \cite{zhang2024evaluating}.
\begin{lemma}
\label{lemma:martingale-equicontinuity}

Let
$
M_n f
=
\frac{1}{n}\sum_{i=1}^n
\left[
f(O_i)-E\{f(O_i)\mid\mathcal F_{i-1}\}
\right]$.
Let $\Theta$ be a index set 
and let
\(\{f_\theta:\theta\in\Theta\}\)
be a uniformly bounded class satisfying
$
\|f_\theta-f_{\theta'}\|_\infty
\leq
L\|\theta-\theta'\|_\infty$.
Define the entropy integral as 
$
J_\infty(\epsilon,\Theta)
=
\int_0^\epsilon
\sqrt{
\log\left\{
1+N_\infty(u,\Theta)
\right\}
}
\,du.
$
Suppose that \(J_\infty(\epsilon,\Theta)<\infty\) is finite and that
\(J_\infty(\epsilon,\Theta)\to0\) as \(\epsilon \to 0\).
Let \(M=\sup_{\theta\in\Theta}\|\theta\|_\infty\),
\(Z=8LM\), and define
$
s_n(f_n)
=
\left\{
\frac{1}{n}\sum_{i=1}^n
E\left[
\rho\left(
\frac{|f_n(O_i)|}{Z}
\right)
\Bigm|
\mathcal F_{i-1}
\right]
\right\}^{1/2}$,
where $\rho(x) = e^x-x-1$.

If \(s_n(f_n)=O_P(\delta_n)\) for some
\(\delta_n\to0\), then
$M_n f_n=o_P(n^{-1/2})$.
\end{lemma}

\begin{theorem}[Equicontinuity of \(M_{2,n}^{\bd}\) and
\(M_{3,n}^{\bd}\)]
\label{theorem:longitudinal-equicontinuity}

Suppose
Assumptions~\ref{assumption:longitudinal-average-reference-positivity},
\ref{assumption:longitudinal-equicontinuity-conditions},
\ref{assumption:longitudinal-equicontinuity-convergence-Q}, and
\ref{assumption:longitudinal-equicontinuity-convergence-g}
hold. Then
$
M_{2,n}^{\bd}
=
o_P(n^{-1/2})$, $M_{3,n}^{\bd} = o_P(n^{-1/2})$.
\end{theorem}

\begin{proof}
For \(M_{2,n}^{\bd}\), define
$\Theta_2^{\bd}
=
\mathcal M^{\bd}\times\cG$,
and, for
\(\theta_2=(\mathbf m^{\bd},\bar{\bfg})\in\Theta_2^{\bd}\), define
$
\|\theta_2\|_\infty
=
\max
\left\{
\|\mathbf m^{\bd}\|_\infty,
\|\bar{\bfg}\|_\infty
\right\}.
$

Let
$
\widetilde D^{\bd}(P_{\bfQ,\bar{\bfg}})
= m_1^{\bd}(\bfQ)+\sum_{k=1}^K D_k^{\bd}(P_{\bfQ,\bar{\bfg}}),
$
and define
$
f_{2,\theta_2}^{\bd}(O)
=
\widetilde D^{\bd}(P_{\bfQ,\bar{\bfg}})(O)
-
\widetilde D^{\bd}(P_{\bfQ_\infty,\bar{\bfg}})(O).
$
Thus,
$
M_{2,n}^{\bd}
=
\frac{1}{n}\sum_{i=1}^n
\left[
f_{2,\theta_{2,n}^{\bd}}^{\bd}(O_i)
-
P_{\bfQ_0,\bfg_i}
f_{2,\theta_{2,n}^{\bd}}^{\bd}
\right]$, 
which is a martingale process indexed by 
$\theta_{2,n}^{\bd}
=
\left(
\mathbf m^{\bd}(\bfQ_n^*),
\bar{\bfg}_n
\right)$.
A direct calculation of the longitudinal EIC under Assumption \ref{assumption:longitudinal-average-reference-positivity} yields
$
\left\|
f_{2,\theta_2}^{\bd}
-
f_{2,\theta_2'}^{\bd}
\right\|_\infty
\leq
L_2
\left\|
\theta_2-\theta_2'
\right\|_\infty$,
where $
L_2
=
\frac{K}{\zeta}
+
\frac{K^2}{\zeta^2}$.

For \(M_{3,n}^{\bd}\), define
$
\Theta_3^{\bd}
=
\cG$ and $
f_{3,\bar{\bfg}}^{\bd}(O)
=
D^{\bd}(P_{\bfQ_\infty,\bar{\bfg}})(O)$.
We have that 
\[
M_{3,n}^{\bd}
=
\frac{1}{n}\sum_{i=1}^n
\left[
\{f_{3,\bar{\bfg}_n}^{\bd}
-
f_{3,\bar{\bfg}_\infty}^{\bd}\}(O_i)
-
P_{\bfQ_0,\bfg_i}
\{f_{3,\bar{\bfg}_n}^{\bd}
-
f_{3,\bar{\bfg}_\infty}^{\bd}\}
\right],
\]
where 
$
\left\|
f_{3,\bar{\bfg}}^{\bd}
-
f_{3,\bar{\bfg}'}^{\bd}
\right\|_\infty
\leq
L_3
\left\|
\bar{\bfg}-\bar{\bfg}'
\right\|_\infty$,
where 
$L_3=\frac{K(K+1)}{2\zeta^2}.
$

Moreover, since 
$
N_\infty(u,\Theta_2^{\bd})
\leq
N_\infty(u,\mathcal M^{\bd})
N_\infty(u,\cG)$, and $
N_\infty(u,\Theta_3^{\bd})
=
N_\infty(u,\cG)
$, Assumption~\ref{assumption:longitudinal-equicontinuity-conditions}
implies the required entropy conditions. Since both outcome regressions and treatment probabilities are bounded by 1, we have 
$c_2=\frac{8K}{\zeta}
+
\frac{8K^2}{\zeta^2}$ and $c_3=\frac{4K(K+1)}{\zeta^2}$.

Finally,
Assumptions~\ref{assumption:longitudinal-equicontinuity-convergence-Q}
and~\ref{assumption:longitudinal-equicontinuity-convergence-g}
provide the required convergence in Lemma \ref{lemma:martingale-equicontinuity}. Applying the
general martingale equicontinuity result separately to the two processes yields
$M_{2,n}^{\bd}
=
o_P(n^{-1/2})$ and $
M_{3,n}^{\bd}
=
o_P(n^{-1/2})$.

\end{proof}

Combining
Theorems~\ref{theorem:longitudinal-tmle-decomposition},
\ref{theorem:longitudinal-asymptotic-normality-M1}, and
\ref{theorem:longitudinal-equicontinuity}
yields the asymptotic normality of ADL-LTMLE \ref{theorem:longitudinal-asymptotic-normality}:
$
\sqrt n
\left\{
\psi^{\bd}(\bfQ_n^*)
-
\psi^{\bd}(\bfQ_0)
\right\}
\dto
N(0,\sigma_{\bd}^2)$.

\subsection{Semiparametric Efficiency}
\label{subsec:longitudinal-efficiency}

The semiparametric efficiency analysis of the ADL-LTMLE builds directly on the ADL-TMLE framework of \citet{zhang2025efficient}, which establishes efficiency through local asymptotic normality and the convolution theorem \cite{van1996weak}. 

Write
$
q(o)
=
q_1(l_1)
\prod_{k=1}^K
q_{k+1}(l_{k+1}\mid a_k,h_k)$
and
$
g_i(o)
=
\prod_{k=1}^K
g_{k,i}(a_k\mid h_k)$.
Then the longitudinal adaptive likelihood has the general factorization
$
p_{\bfQ,\bfg_{1:n}}^n(o_1,\ldots,o_n)
=
\prod_{i=1}^n q(o_i)g_i(o_i).
$
Define a submodel
\(\{\bfQ_\epsilon^h:\epsilon\}\) through \(\bfQ_0\) by
\[
dQ_{1,\epsilon}^h(l_1)
=
\{1+\epsilon h_1(l_1)\}
dQ_{0,1}(l_1)
\]
and
\[
dQ_{k+1,\epsilon}^h
(l_{k+1}\mid a_k,h_k)
=
\left\{
1+
\epsilon
h_{k+1}(l_{k+1},a_k,h_k)
\right\}
dQ_{0,k+1}
(l_{k+1}\mid a_k,h_k)
\]
for \(k\in[K]\).

Let
$
P_{\bfQ_\epsilon^h,\bfg_{1:n}}^n$ denote the joint law of
\((O_1,\ldots,O_n)\) under 
\(\bfQ_\epsilon^h\) and \(\bfg_{1:n}\). Its density is factorized by 
\[
\begin{aligned}
&
p_{\bfQ_\epsilon^h,\bfg_{1:n}}^n(o_1,\ldots,o_n)
\\
&\qquad=
\prod_{i=1}^n
\left[
q_{1,\epsilon}^h(l_{1,i})
\prod_{k=1}^K
g_{k,i}
(a_{k,i}\mid h_{k,i},\mathcal F_{i-1})
q_{k+1,\epsilon}^h
(l_{k+1,i}\mid a_{k,i},h_{k,i})
\right].
\end{aligned}
\]

\begin{proposition}[Local Asymptotic Normality]
\label{proposition:longitudinal-LAN}

Under
Assumption~\ref{assumption:longitudinal-hilbert-norm},
for every bounded score direction \(h\) and fixed \(t\in\mathbb R\),
\[
\log
\frac{
dP_{\bfQ_{t/\sqrt n}^{h},\bfg_{1:n}}^n
}{
dP_{\bfQ_0,\bfg_{1:n}}^n
}
=
t\Delta_{n,h}
-
\frac{t^2}{2}
P_{\bfQ_0,\bar{\bfg}_\infty}h^2
+
o_P(1),
\]
where
$
\Delta_{n,h}
=
\frac1{\sqrt n}\sum_{i=1}^n h(O_i)
\dto
N
\left(
0,
P_{\bfQ_0,\bar{\bfg}_\infty}h^2
\right)$.
\end{proposition}

\begin{proof}[Proof of
Proposition~\ref{proposition:longitudinal-LAN}]

For participant \(i\), write
$
h_{1,i}=h_1(L_{1,i})$
and, for \(k\in[K]\),
$
h_{k+1,i}
=
h_{k+1}(L_{k+1,i},A_{k,i},H_{k,i})$. Therefore, $h(O_i)=\sum_{k=1}^{K+1}h_{k,i}$.

The log-likelihood ratio can be as 
\[
\log
\frac{
dP_{\bfQ_{t/\sqrt n}^{h},\bfg_{1:n}}^n
}{
dP_{\bfQ_0,\bfg_{1:n}}^n
}
=
\sum_{i=1}^n
\sum_{k=1}^{K+1}
\log
\left(
1+\frac{t}{\sqrt n}h_{k,i}
\right).
\]

Applying Taylor expansion, we have
\[
\begin{aligned}
\log
\frac{
dP_{\bfQ_{t/\sqrt n}^{h},\bfg_{1:n}}^n
}{
dP_{\bfQ_0,\bfg_{1:n}}^n
}
=
\frac{t}{\sqrt n}
\sum_{i=1}^n
\sum_{k=1}^{K+1}h_{k,i}
-
\frac{t^2}{2n}
\sum_{i=1}^n
\sum_{k=1}^{K+1}h_{k,i}^2
+
o(1).
\end{aligned}
\]
Define $P_{\bfQ_0,\bfg_i} f = E\{f(O_i)\mid\mathcal F_{i-1}\}$.
Since the longitudinal score components are conditionally orthogonal, we have $P_{\bfQ_0,\bfg_i} h=0$ and 
$
P_{\bfQ_0,\bfg_i}(h_k h_s)
=
0$ for $k \neq s$.
Consequently,
$
P_{\bfQ_0,\bfg_i} h^2
=
\sum_{k=1}^{K+1}P_{\bfQ_0,\bfg_i} h_k^2$.

Define
$
Z_{i,h}
=
\sum_{k=1}^{K+1}h_{k,i}^2
-
P_{\bfQ_0,\bfg_i}
\left(
\sum_{k=1}^{K+1}h_k^2
\right)$.
The variables \(Z_{i,h}\) form a bounded martingale-difference sequence.
Therefore,
$
\frac1n\sum_{i=1}^n Z_{i,h}
=
o_P(1)$.
Under Assumption~\ref{assumption:longitudinal-hilbert-norm}, 
$
\frac1n
\sum_{i=1}^n
\sum_{k=1}^{K+1}h_{k,i}^2
=
P_{\bfQ_0,\bar{\bfg}_n}h^2
+
o_P(1)
\pto
P_{\bfQ_0,\bar{\bfg}_\infty}h^2
$.
Since the conditional Lindeberg condition holds with $h$ bounded,
the martingale CLT gives
$
\Delta_{n,h}
\dto
N
\left(
0,
P_{\bfQ_0,\bar{\bfg}_\infty}h^2
\right)$, which proves the
proposition.
\end{proof}

For the submodel
\(\{\bfQ_\epsilon^h:\epsilon\}\)
defined above, let
$
\dot\psi_{\bfQ_0}^{\bd}(h)
:=
\left.
\frac{d}{d\epsilon}
\psi^{\bd}(\bfQ_\epsilon^h)
\right|_{\epsilon=0}$
denote the pathwise derivative of
\(\psi^{\bd}\) at \(\bfQ_0\) in the score direction \(h\).
Under the fixed longitudinal reference law
\(P_{\bfQ_0,\bar{\bfg}_\infty}\), the pathwise derivative satisfies
$
\dot\psi_{\bfQ_0}^{\bd}(h)
=
P_{\bfQ_0,\bar{\bfg}_\infty}
\left(
D^{\bd} (P_{\bfQ_0,\bar{\bfg}_\infty})h
\right)$, for any $h\in\mathcal H(\bfQ_0,\bar{\bfg}_\infty)$.
Thus, \(D^{\bd} (P_{\bfQ_0,\bar{\bfg}_\infty})\) is the canonical gradient in the limiting reference experiment $P_{\bfQ_0,\bar{\bfg}_\infty}$. The corresponding semiparametric
efficiency bound is therefore 
$P_{\bfQ_0,\bar{\bfg}_\infty}
\left[
\left(
D^{\bd} (P_{\bfQ_0,\bar{\bfg}_\infty})
\right)^2
\right]$.

When \(\bfQ_\infty=\bfQ_0\), the
difference $\psi^{\bd}(\bfQ_n^*)
-
\psi^{\bd}(\bfQ_0)$
has \(\frac{1}{n}\sum_{i=1}^n
\left[
D^{\bd}(P_{\bfQ_0,\bar{\bfg}_\infty})(O_i)
-
P_{\bfQ_0,\bfg_i}
D^{\bd}(P_{\bfQ_0,\bar{\bfg}_\infty})
\right]\) as the leading term. Since $P_{\bfQ_0,\bfg_i}
D^{\bd}(P_{\bfQ_0,\bar{\bfg}_\infty})
=
0$ and Assumption \ref{assumption:longitudinal-hilbert-norm} hold, we have 
\begin{align*}
\frac{1}{n}\sum_{i=1}^n
\Var_{P_{\bfQ_0,\bfg_i}}
D^{\bd} (P_{\bfQ_0,\bar{\bfg}_\infty})
=
P_{\bfQ_0,\bar{\bfg}_n}
\left[
\left(
D^{\bd} (P_{\bfQ_0,\bar{\bfg}_\infty})
\right)^2
\right]
\pto
P_{\bfQ_0,\bar{\bfg}_\infty}
\left[
\left(
D^{\bd} (P_{\bfQ_0,\bar{\bfg}_\infty})
\right)^2
\right].
\end{align*}
Thus, the limiting variance of ADL-LTMLE equals the squared norm of the canonical gradient under the limiting reference law. 

The regularity arguments developed for the
ADL-TMLE in \citet{zhang2025efficient} directly applies to the present
longitudinal experiment using 
the corresponding longitudinal stage-wise score.
By the
convolution theorem \citep{van1996weak}, the asymptotic distribution
of any regular estimator is the convolution of a normal distribution 
$N\left(
0,
P_{\bfQ_0,\bar{\bfg}_\infty}
\left[
\left\{
D^{\bd}
(P_{\bfQ_0,\bar{\bfg}_\infty})
\right\}^2
\right]
\right)
$
with an independent noise distribution. Since the ADL-LTMLE is regular and
attains this Gaussian limit without additional noise, it is semiparametric
efficient.

\section{Data-generating mechanisms}
\label{supp_sec:dgd}
\paragraph{Two-stage setting ($K=2$).}
For the two-stage setting, the observed data are
$
O = (L_1, A_1, L_2, A_2, Y)$, 
where $L_1,L_2 \in \{0,1\}$ and $Y$ is continuous.
The baseline covariate is generated as $L_1 \sim \mathrm{Bernoulli}(0.5)$.
Conditional on $(L_1,A_1)$, the intermediate covariate $L_2$ satisfies
$
P(L_2 = 1 \mid L_1=0, A_1=0)=0.2$, $
P(L_2 = 1 \mid L_1=0, A_1=1)=0.8$,
$
P(L_2 = 1 \mid L_1=1, A_1=0)=0.7$, and $
P(L_2 = 1 \mid L_1=1, A_1=1)=0.3.
$

The outcome $Y$ is generated as
$
Y \mid (L_1,A_1,L_2,A_2) \sim
\mathcal{N}\!\left(\mu(L_1,A_1,L_2,A_2),\;\sigma^2(L_1,A_1,A_2)\right)$,
with
$
\mu(L_1,A_1,L_2,A_2)
=
8 + 2L_1 + 1.5L_2
+ (4 + 2(1-L_1))A_1
+ (2 + 3L_1)A_2 
- 1.5 A_1A_2
+ 1.5 L_2A_2 .
$
The standard deviation $\sigma(L_1,A_1,A_2)$ is set to $0.5$ when $(L_1,A_1,A_2)=(0,0,0)$, to $4$ when $(L_1,A_1,A_2)=(0,0,1)$
or $(1,1,1)$, to $1$ when
$(L_1,A_1,A_2)=(1,0,0)$, and to $1$ otherwise.

\paragraph{Three-stage setting ($K=3$).}
For the three-stage setting, the observed data are
$
O = (L_1, A_1, L_2, A_2, L_3, A_3, Y)$,
where $L_1,L_2,L_3 \in \{0,1\}$ and $Y$ is continuous.
The baseline covariate is generated as $L_1 \sim \mathrm{Bernoulli}(0.5)$.
Conditional on $(L_1,A_1)$,
$
P(L_2 = 1 \mid L_1,A_1)
=
\min\!\left\{0.9,\max\!\left(0.1,\; 0.3 + 0.1L_1 + 0.2A_1\right)\right\}.
$
Given $(L_1,A_1,L_2,A_2)$, 
$
P(L_3 = 1 \mid L_1,A_1,L_2,A_2)
=
\min\left\{0.9,\max\left(0.1,\; 0.3 + 0.1L_2 + 0.1A_2 + 0.05A_1\right)\right\}.
$

The outcome regression of $Y$ is generated as
$
Y \mid (L_1,A_1,L_2,A_2,L_3,A_3)
\sim \mathcal{N}(\mu,\sigma^2)$,
with
$
\mu
=
10 + 2L_1 + 2L_2 + 2L_3
+ 3A_1 + 2A_2 + A_3
- A_1A_2 - 0.5 A_2A_3$,
and the standard deviation $\sigma$ equal to $4$ when $(A_1,A_2,A_3)=(0,0,1)$,
equal to $2$ when $(A_1,A_2,A_3)=(0,1,1)$, and equal to $1$ otherwise.

\section{Additional Simulation Results}
We also implemented one-step estimators under the simulation scenarios in Section~\ref{sec:simulation}, as analogues of the ADL-LTMLE
and AD-LTMLE estimators. Specifically, they are defined by:
$
\hat\psi_{n,\mathrm{one\text{-}step}}^\tau
=
\psi_{n,\mathrm{init}}^\tau(\bfQ_n)
+
\frac{1}{n}\sum_{i=1}^n
D^\tau\!\left(P_{\bfQ_n,\bar{\bfg}_n}\right)(O_i)$ and $
\hat\psi_{n,\mathrm{one\text{-}step}}^{\tau,\dagger}
=
\psi_{n,\mathrm{init}}^\tau(\bfQ_n)
+
\frac{1}{n}\sum_{i=1}^n
D^\tau\!\left(P_{\bfQ_n,\bfg_i}\right)(O_i)$, respectively, where $\bfQ_n$ denotes the initial estimate of $\bfQ$. 
The two estimators coincide under a fixed design. The results are reported in Tables~\ref{tab:K2_bar_g_onestep} and Tables~\ref{tab:K3_bar_g_onestep} for $K=2$ and $K=3$, respectively, which show performance similar to that of the corresponding LTMLE-based estimators.

\input{table_K2_Onestep}
\input{table_K3_Onestep}

Next, we conduct sensitivity analyses for the ADL-LTMLE approach. We vary
the total and cohort-wise sample size, and report results using both the
correctly specified strata-specific mean estimator and misspecified working
outcome-regression models.
Results across Tables \ref{tab:k2_ltmle_sample_size_sensitivity}-\ref{tab:k3_ltmle_sample_size_misspecification_sensitivity} show that TarVar decreases substantially as the total sample size
increases for both $K=2$ and $K=3$.
While misspecified outcome regressions increase variance, coverage remains near nominal level across different sample sizes, supporting the robustness of ADL-LTMLE to model misspecification.

\input{table_K2_ltmle_sample_size_sensitivity}

\input{table_K3_ltmle_sample_size_sensitivity}

\lhead[\footnotesize\thepage\fancyplain{}\leftmark]{}\rhead[]{\fancyplain{}\rightmark\footnotesize\thepage}

\bibliography{main}

\end{document}

%% file: notations.tex
\newcommand{\bd}{\mathbf{d}}
\newcommand{\bO}{\bar{O}}

\newcommand{\cA}{\mathcal{A}}

\newcommand{\cH}{\mathcal{H}}

\newcommand{\indep}{\perp\!\!\!\!\perp} 

\newcommand{\bfQ}{\mathbf{Q}} 
\newcommand{\bfg}{\mathbf{g}} 

\newcommand{\I}{I} 
\newcommand{\Var}{Var}

\newcommand{\cQ}{\mathcal{Q}}
\newcommand{\cG}{\mathcal{G}}

\def\pto{\overset{p}{\to}}
\def\dto{\overset{d}{\to}}

%% file: table_K2_LTMLE.tex
\begin{table}[tbh!]
\caption{\small Simulation results for different longitudinal designs and estimators for $K=2$.
We report the empirical variance of the estimates of $\psi^{1}$ (Var$_1$) and $\psi^{2}$ (Var$_2$), the weighted sum of variances using user-specified weight $w$ (TarVar), the empirical $95\%$ coverage probabilities (Cov$_1$, Cov$_2$), and the corresponding oracle coverage probabilities (OCov$_1$, OCov$_2$).}
\label{tab:K2_bar_g_ltmle}\par
\centering
\resizebox{0.95\linewidth}{!}{
\begin{tabular}{|lll rrr rrrr|}
\hline
$w$ & Design & Estimator
& TarVar & Var$_1$ & Var$_2$
& Cov$_1$(\%) & Cov$_2$(\%) & OCov$_1$(\%) & OCov$_2$(\%) \\
\hline
\multirow{4}{*}{$(1/2,\,1/2)$}
& Non-adaptive
& LTMLE
& 0.029 & 0.040 & 0.018 & 95.0 & 97.0 & 95.4 & 95.6\\

& Adaptive
& AD-LTMLE
& 0.022 & 0.030 & 0.013 & 92.6 & 94.4 & 94.8 & 94.2\\

& Adaptive
& ADL-LTMLE
& 0.019 & 0.026 & 0.012 & 93.0 & 94.6 & 94.2 & 95.4\\

& Oracle
& LTMLE
& 0.017 & 0.024 & 0.011 & 94.2 & 94.6 & 95.0 & 95.0\\
\hline
\multirow{4}{*}{$(1/3,\,2/3)$}
& Non-adaptive
& LTMLE
& 0.026 & 0.040 & 0.018 & 95.0 & 97.0 & 95.4 & 95.6\\

& Adaptive
& AD-LTMLE
& 0.019 & 0.032 & 0.013 & 92.0 & 96.0 & 94.2 & 96.6\\

& Adaptive
& ADL-LTMLE
& 0.017 & 0.028 & 0.012 & 92.6 & 95.6 & 95.0 & 96.0\\

& Oracle
& LTMLE
& 0.015 & 0.024 & 0.010 & 94.8 & 95.6 & 95.4 & 95.6\\
\hline
\multirow{4}{*}{$(2/3,\,1/3)$}
& Non-adaptive
& LTMLE
& 0.033 & 0.040 & 0.018 & 95.0 & 97.0 & 95.4 & 95.6\\

& Adaptive
& AD-LTMLE
& 0.025 & 0.030 & 0.014 & 92.2 & 94.4 & 95.4 & 95.2\\

& Adaptive
& ADL-LTMLE
& 0.021 & 0.026 & 0.012 & 93.2 & 94.8 & 95.4 & 95.6\\

& Oracle
& LTMLE
& 0.019 & 0.023 & 0.012 & 94.6 & 94.6 & 95.2 & 94.8\\
\hline
\end{tabular}
}
\end{table}

%% file: table_K3_LTMLE.tex
\begin{table}[tbh!]
\caption{\small Simulation results for different designs amd estimators for $K=3$.
We report the empirical variance of the estimates of
$\psi^{1}$ (Var$_1$), $\psi^{2}$ (Var$_2$), and $\psi^{3}$ (Var$_3$),
the weighted sum of variances using user-specified weight $w$ (TarVar), the empirical $95\%$ coverage probabilities (Cov$_1$, Cov$_2$, Cov$_3$),
and the corresponding oracle coverage probabilities
(OCov$_1$, OCov$_2$, OCov$_3$).}
\label{tab:K3_bar_g_ltmle}\par
\centering
\resizebox{0.95\linewidth}{!}{
\begin{tabular}{|lll rrrr rrr rrr|}
\hline
$w$ & Design & Estimator
& TarVar & Var$_1$ & Var$_2$ & Var$_3$
& Cov$_1$(\%) & Cov$_2$(\%) & Cov$_3$(\%)
& OCov$_1$(\%) & OCov$_2$(\%) & OCov$_3$(\%) \\
\hline
\multirow{4}{*}{$(1/3,\,1/3,\,1/3)$}
& Non-adaptive
& LTMLE
& 0.056 & 0.026 & 0.079 & 0.063 & 94.2 & 94.6 & 96.2 & 94.0 & 94.2 & 95.4 \\

& Adaptive
& AD-LTMLE
& 0.031 & 0.022 & 0.039 & 0.033 & 95.2 & 94.6 & 94.6 & 95.4 & 94.8 & 94.0 \\

& Adaptive
& ADL-LTMLE
& 0.026 & 0.020 & 0.031 & 0.026 & 95.2 & 94.4 & 95.2 & 95.2 & 95.0 & 95.4 \\

& Oracle
& LTMLE
& 0.022 & 0.020 & 0.025 & 0.020 & 95.0 & 94.8 & 96.4 & 94.8 & 94.0 & 94.6 \\
\hline
\multirow{4}{*}{$(1/6,\,1/3,\,1/2)$}
& Non-adaptive
& LTMLE
& 0.062 & 0.026 & 0.079 & 0.063 & 94.2 & 94.6 & 96.2 & 94.0 & 94.2 & 95.4 \\

& Adaptive
& AD-LTMLE
& 0.032 & 0.027 & 0.038 & 0.030 & 94.4 & 95.4 & 94.2 & 94.4 & 95.4 & 94.8 \\

& Adaptive
& ADL-LTMLE
& 0.024 & 0.024 & 0.028 & 0.022 & 94.4 & 95.6 & 95.2 & 94.8 & 95.0 & 94.8 \\

& Oracle
& LTMLE
& 0.021 & 0.026 & 0.024 & 0.017 & 95.8 & 95.8 & 97.2 & 95.4 & 94.4 & 96.2 \\
\hline
\multirow{4}{*}{$(1/2,\,1/3,\,1/6)$}
& Non-adaptive
& LTMLE
& 0.050 & 0.026 & 0.079 & 0.063 & 94.2 & 94.6 & 96.2 & 94.0 & 94.2 & 95.4 \\

& Adaptive
& AD-LTMLE
& 0.029 & 0.019 & 0.039 & 0.038 & 95.0 & 93.4 & 93.0 & 95.0 & 93.6 & 94.4 \\

& Adaptive
& ADL-LTMLE
& 0.024 & 0.017 & 0.031 & 0.029 & 94.8 & 94.2 & 94.6 & 95.4 & 94.4 & 95.4 \\

& Oracle
& LTMLE
& 0.021 & 0.016 & 0.026 & 0.025 & 96.2 & 95.6 & 95.2 & 95.0 & 95.0 & 94.4 \\
\hline
\end{tabular}
}
\end{table}

%% file: table_K2_Onestep.tex
\begin{table}[tbh!]
\caption{\small Simulation results for $K=2$ using one-step estimators.
Performance of one-step estimators under different longitudinal designs.
We report the empirical variance of the estimates of $\psi^{1}$ (Var$_1$) and $\psi^{2}$ (Var$_2$), the weighted sum of variances using normalized user-specified weight $w$ (TarVar), the empirical $95\%$ coverage probabilities (Cov$_1$, Cov$_2$), and the corresponding oracle coverage probabilities (OCov$_1$, OCov$_2$) across 500 Monte Carlo runs.}
\label{tab:K2_bar_g_onestep}\par
\centering
\resizebox{0.95\linewidth}{!}{
\begin{tabular}{|lll rrr rrrr|}
\hline
$w$ & Design & Estimator
& TarVar & Var$_1$ & Var$_2$
& Cov$_1$(\%) & Cov$_2$(\%) & OCov$_1$(\%) & OCov$_2$(\%) \\
\hline
\multirow{4}{*}{$(1/2,\,1/2)$}
& Non-adaptive
& $\hat\psi_{n,\mathrm{one-step}}^{\tau}$
& 0.029 & 0.040 & 0.018 & 95.0 & 97.0 & 95.4 & 95.6\\

& Adaptive
& $\hat\psi_{n,\mathrm{one-step}}^{\tau,\dagger}$
& 0.022 & 0.030 & 0.013 & 92.8 & 94.4 & 94.8 & 94.4\\

& Adaptive
& $\hat\psi_{n,\mathrm{one-step}}^{\tau}$
& 0.019 & 0.026 & 0.012 & 93.0 & 94.6 & 94.2 & 95.4\\

& Oracle
& $\hat\psi_{n,\mathrm{one-step}}^{\tau}$
& 0.017 & 0.024 & 0.011 & 94.2 & 94.6 & 95.0 & 95.0\\
\hline
\multirow{4}{*}{$(1/3,\,2/3)$}
& Non-adaptive
& $\hat\psi_{n,\mathrm{one-step}}^{\tau}$
& 0.026 & 0.040 & 0.018 & 95.0 & 97.0 & 95.4 & 95.6\\

& Adaptive
& $\hat\psi_{n,\mathrm{one-step}}^{\tau,\dagger}$
& 0.019 & 0.032 & 0.013 & 91.6 & 95.8 & 94.4 & 96.6\\

& Adaptive
& $\hat\psi_{n,\mathrm{one-step}}^{\tau}$
& 0.017 & 0.028 & 0.012 & 92.6 & 95.6 & 95.0 & 96.0\\

& Oracle
& $\hat\psi_{n,\mathrm{one-step}}^{\tau}$
& 0.015 & 0.024 & 0.010 & 94.8 & 95.6 & 95.4 & 95.6\\
\hline
\multirow{4}{*}{$(2/3,\,1/3)$}
& Non-adaptive
& $\hat\psi_{n,\mathrm{one-step}}^{\tau}$
& 0.033 & 0.040 & 0.018 & 95.0 & 97.0 & 95.4 & 95.6\\

& Adaptive
& $\hat\psi_{n,\mathrm{one-step}}^{\tau,\dagger}$
& 0.025 & 0.030 & 0.014 & 92.2 & 94.4 & 95.4 & 95.2\\

& Adaptive
& $\hat\psi_{n,\mathrm{one-step}}^{\tau}$
& 0.021 & 0.026 & 0.012 & 93.2 & 94.8 & 95.4 & 95.6\\

& Oracle
& $\hat\psi_{n,\mathrm{one-step}}^{\tau}$
& 0.019 & 0.023 & 0.012 & 94.6 & 94.6 & 95.2 & 94.8\\
\hline
\end{tabular}
}
\end{table}

%% file: table_K3_Onestep.tex
\begin{table}[tbh!]
\caption{\small Simulation results for $K=3$ (one-step estimator).
Performance of one-step estimators under different longitudinal designs.
We report the weighted sum of variances using normalized user-specified weight $w$ (TarVar),
the empirical variance of the estimates of
$\psi^{1}$ (Var$_1$), $\psi^{2}$ (Var$_2$), and $\psi^{3}$ (Var$_3$),
the empirical $95\%$ coverage probabilities (Cov$_1$, Cov$_2$, Cov$_3$),
and the corresponding oracle coverage probabilities
(OCov$_1$, OCov$_2$, OCov$_3$) across 500 Monte Carlo runs.}
\label{tab:K3_bar_g_onestep}\par
\centering
\resizebox{0.95\linewidth}{!}{
\begin{tabular}{|lll rrrr rrr rrr|}
\hline
$w$ & Design & Estimator
& TarVar & Var$_1$ & Var$_2$ & Var$_3$
& Cov$_1$(\%) & Cov$_2$(\%) & Cov$_3$(\%)
& OCov$_1$(\%) & OCov$_2$(\%) & OCov$_3$(\%) \\
\hline
\multirow{4}{*}{$(1/3,\,1/3,\,1/3)$}
& Non-adaptive
& $\hat\psi_{n,\mathrm{one\text{-}step}}^{\tau}$
& 0.056 & 0.026 & 0.079 & 0.063 & 94.2 & 94.6 & 96.2 & 94.0 & 94.2 & 95.4 \\

& Adaptive
& $\hat\psi_{n,\mathrm{one\text{-}step}}^{\tau,\dagger}$
& 0.031 & 0.022 & 0.039 & 0.033 & 95.2 & 94.6 & 94.6 & 95.4 & 94.6 & 93.6 \\

& Adaptive
& $\hat\psi_{n,\mathrm{one\text{-}step}}^{\tau}$
& 0.026 & 0.020 & 0.031 & 0.026 & 95.2 & 94.4 & 95.2 & 95.2 & 95.0 & 95.4 \\

& Oracle
& $\hat\psi_{n,\mathrm{one\text{-}step}}^{\tau}$
& 0.022 & 0.020 & 0.025 & 0.020 & 95.0 & 94.8 & 96.4 & 94.8 & 94.0 & 94.6 \\
\hline
\multirow{4}{*}{$(1/6,\,1/3,\,1/2)$}
& Non-adaptive
& $\hat\psi_{n,\mathrm{one\text{-}step}}^{\tau}$
& 0.062 & 0.026 & 0.079 & 0.063 & 94.2 & 94.6 & 96.2 & 94.0 & 94.2 & 95.4 \\

& Adaptive
& $\hat\psi_{n,\mathrm{one\text{-}step}}^{\tau,\dagger}$
& 0.032 & 0.027 & 0.038 & 0.029 & 94.6 & 95.2 & 94.4 & 94.6 & 95.0 & 94.6 \\

& Adaptive
& $\hat\psi_{n,\mathrm{one\text{-}step}}^{\tau}$
& 0.024 & 0.024 & 0.028 & 0.022 & 94.4 & 95.6 & 95.2 & 94.8 & 95.0 & 94.8 \\

& Oracle
& $\hat\psi_{n,\mathrm{one\text{-}step}}^{\tau}$
& 0.021 & 0.026 & 0.024 & 0.017 & 95.8 & 95.8 & 97.2 & 95.4 & 94.4 & 96.2 \\
\hline
\multirow{4}{*}{$(1/2,\,1/3,\,1/6)$}
& Non-adaptive
& $\hat\psi_{n,\mathrm{one\text{-}step}}^{\tau}$
& 0.050 & 0.026 & 0.079 & 0.063 & 94.2 & 94.6 & 96.2 & 94.0 & 94.2 & 95.4 \\

& Adaptive
& $\hat\psi_{n,\mathrm{one\text{-}step}}^{\tau,\dagger}$
& 0.029 & 0.019 & 0.039 & 0.038 & 95.0 & 93.4 & 93.2 & 95.0 & 93.2 & 93.8 \\

& Adaptive
& $\hat\psi_{n,\mathrm{one\text{-}step}}^{\tau}$
& 0.024 & 0.017 & 0.031 & 0.029 & 94.8 & 94.2 & 94.6 & 95.4 & 94.4 & 95.4 \\

& Oracle
& $\hat\psi_{n,\mathrm{one\text{-}step}}^{\tau}$
& 0.021 & 0.016 & 0.026 & 0.025 & 96.2 & 95.6 & 95.2 & 95.0 & 95.0 & 94.4 \\
\hline
\end{tabular}
}
\end{table}

%% file: table_K2_ltmle_sample_size_sensitivity.tex
\begin{table}[tbh!]
\caption{\small Complete results of ADL-LTMLE using the correct strata-specific mean estimates ($K=2$).
The settings vary total sample size $n$ and cohort size $b$.}
\label{tab:k2_ltmle_sample_size_sensitivity}\par
\centering
\resizebox{0.98\linewidth}{!}{
\begin{tabular}{|llrrrrrrr|}
\hline
$w$ & Setting & TarVar & Var$_1$ & Var$_2$ & Cov$_1$(\%) & Cov$_2$(\%) & OCov$_1$(\%) & OCov$_2$(\%) \\
\hline
$(1/2,\,1/2)$ & $n=1000,\ b=200$ & 0.036 & 0.049 & 0.024 & 94.6 & 95.4 & 95.4 & 96.0 \\
$(1/2,\,1/2)$ & $n=2000,\ b=200$ & 0.017 & 0.022 & 0.011 & 94.8 & 95.8 & 94.4 & 95.8 \\
$(1/2,\,1/2)$ & $n=2000,\ b=400$ & 0.019 & 0.026 & 0.012 & 93.0 & 94.6 & 94.2 & 95.4 \\
$(1/2,\,1/2)$ & $n=4000,\ b=800$ & 0.008 & 0.011 & 0.006 & 95.8 & 96.4 & 95.0 & 96.0 \\
\hline
$(1/3,\,2/3)$ & $n=1000,\ b=200$ & 0.033 & 0.051 & 0.023 & 93.6 & 94.8 & 94.2 & 95.2 \\
$(1/3,\,2/3)$ & $n=2000,\ b=200$ & 0.015 & 0.023 & 0.011 & 95.2 & 95.0 & 94.2 & 95.0 \\
$(1/3,\,2/3)$ & $n=2000,\ b=400$ & 0.017 & 0.028 & 0.012 & 92.6 & 95.6 & 95.0 & 96.0 \\
$(1/3,\,2/3)$ & $n=4000,\ b=800$ & 0.007 & 0.012 & 0.005 & 96.6 & 94.2 & 96.0 & 94.0 \\
\hline
$(2/3,\,1/3)$ & $n=1000,\ b=200$ & 0.039 & 0.046 & 0.025 & 94.6 & 94.8 & 95.4 & 95.2 \\
$(2/3,\,1/3)$ & $n=2000,\ b=200$ & 0.018 & 0.021 & 0.012 & 95.0 & 95.6 & 94.6 & 95.8 \\
$(2/3,\,1/3)$ & $n=2000,\ b=400$ & 0.021 & 0.026 & 0.012 & 93.2 & 94.8 & 95.4 & 95.6 \\
$(2/3,\,1/3)$ & $n=4000,\ b=800$ & 0.009 & 0.011 & 0.006 & 96.6 & 97.2 & 95.6 & 97.4 \\
\hline
\end{tabular}
}
\end{table}

\begin{table}[tbh!]
\caption{\small Complete results of ADL-LTMLE using the misspecified working outcome regression $Y \sim L_1 + A_1 + L_2 + A_2,
$ instead of the correct outcome model ($K=2$). The settings vary total sample size $n$ and cohort size $b$.}
\label{tab:k2_ltmle_sample_size_misspecification_sensitivity}\par
\centering
\resizebox{0.98\linewidth}{!}{
\begin{tabular}{|llrrrrrrr|}
\hline
$w$ & Setting & TarVar & Var$_1$ & Var$_2$ & Cov$_1$(\%) & Cov$_2$(\%) & OCov$_1$(\%) & OCov$_2$(\%) \\
\hline
$(1/2,\,1/2)$ & $n=1000,\ b=200$ & 0.041 & 0.049 & 0.033 & 94.4 & 93.4 & 95.0 & 94.2 \\
$(1/2,\,1/2)$ & $n=2000,\ b=200$ & 0.019 & 0.022 & 0.017 & 94.8 & 95.6 & 94.6 & 95.4 \\
$(1/2,\,1/2)$ & $n=2000,\ b=400$ & 0.021 & 0.026 & 0.017 & 92.8 & 95.0 & 94.6 & 95.6 \\
$(1/2,\,1/2)$ & $n=4000,\ b=800$ & 0.010 & 0.011 & 0.008 & 95.2 & 95.0 & 95.0 & 96.0 \\
\hline
$(1/3,\,2/3)$ & $n=1000,\ b=200$ & 0.038 & 0.052 & 0.031 & 93.8 & 93.6 & 95.0 & 94.2 \\
$(1/3,\,2/3)$ & $n=2000,\ b=200$ & 0.018 & 0.023 & 0.015 & 94.6 & 94.6 & 94.4 & 94.8 \\
$(1/3,\,2/3)$ & $n=2000,\ b=400$ & 0.019 & 0.028 & 0.014 & 92.8 & 95.0 & 95.4 & 95.0 \\
$(1/3,\,2/3)$ & $n=4000,\ b=800$ & 0.009 & 0.012 & 0.007 & 96.0 & 94.8 & 95.6 & 95.6 \\
\hline
$(2/3,\,1/3)$ & $n=1000,\ b=200$ & 0.043 & 0.047 & 0.037 & 94.6 & 94.0 & 94.8 & 94.6 \\
$(2/3,\,1/3)$ & $n=2000,\ b=200$ & 0.020 & 0.021 & 0.019 & 94.4 & 94.8 & 94.2 & 95.8 \\
$(2/3,\,1/3)$ & $n=2000,\ b=400$ & 0.024 & 0.026 & 0.018 & 94.0 & 94.2 & 95.4 & 95.2 \\
$(2/3,\,1/3)$ & $n=4000,\ b=800$ & 0.010 & 0.011 & 0.009 & 96.4 & 95.4 & 95.2 & 96.2 \\
\hline
\end{tabular}
}
\end{table}

%% file: table_K3_ltmle_sample_size_sensitivity.tex
\begin{table}[H]
\caption{\small Complete results of ADL-LTMLE using the correct strata-specific mean estimates ($K=3$).
The settings vary total sample size $n$ and cohort size $b$.}
\label{tab:k3_ltmle_sample_size_sensitivity}\par
\centering
\resizebox{0.98\linewidth}{!}{
\begin{tabular}{|llrrrrrrrrrr|}
\hline
$w$ & Setting & TarVar & Var$_1$ & Var$_2$ & Var$_3$ & Cov$_1$(\%) & Cov$_2$(\%) & Cov$_3$(\%) & OCov$_1$(\%) & OCov$_2$(\%) & OCov$_3$(\%) \\
\hline
$(1/3,\,1/3,\,1/3)$ & $n=1000,\ b=200$ & 0.052 & 0.042 & 0.061 & 0.055 & 93.4 & 94.8 & 92.4 & 95.0 & 95.8 & 95.2 \\
$(1/3,\,1/3,\,1/3)$ & $n=2000,\ b=200$ & 0.025 & 0.020 & 0.029 & 0.025 & 94.2 & 95.0 & 93.6 & 95.0 & 95.4 & 95.2 \\
$(1/3,\,1/3,\,1/3)$ & $n=2000,\ b=400$ & 0.026 & 0.020 & 0.031 & 0.026 & 95.2 & 94.4 & 95.2 & 95.2 & 95.0 & 95.4 \\
$(1/3,\,1/3,\,1/3)$ & $n=4000,\ b=800$ & 0.013 & 0.010 & 0.015 & 0.012 & 95.0 & 95.2 & 94.6 & 95.4 & 95.0 & 94.6 \\
\hline
$(1/6,\,1/3,\,1/2)$ & $n=1000,\ b=200$ & 0.052 & 0.049 & 0.060 & 0.048 & 93.4 & 92.8 & 92.2 & 94.6 & 93.2 & 93.8 \\
$(1/6,\,1/3,\,1/2)$ & $n=2000,\ b=200$ & 0.024 & 0.024 & 0.027 & 0.021 & 93.4 & 94.8 & 94.2 & 93.4 & 94.8 & 94.6 \\
$(1/6,\,1/3,\,1/2)$ & $n=2000,\ b=400$ & 0.024 & 0.024 & 0.028 & 0.022 & 94.4 & 95.6 & 95.2 & 94.8 & 95.0 & 94.8 \\
$(1/6,\,1/3,\,1/2)$ & $n=4000,\ b=800$ & 0.012 & 0.012 & 0.014 & 0.011 & 94.6 & 97.0 & 95.8 & 95.0 & 96.8 & 96.2 \\
\hline
$(1/2,\,1/3,\,1/6)$ & $n=1000,\ b=200$ & 0.048 & 0.035 & 0.062 & 0.059 & 95.8 & 96.4 & 93.4 & 95.8 & 96.4 & 94.8 \\
$(1/2,\,1/3,\,1/6)$ & $n=2000,\ b=200$ & 0.023 & 0.017 & 0.029 & 0.029 & 94.8 & 94.8 & 94.6 & 95.2 & 94.2 & 95.8 \\
$(1/2,\,1/3,\,1/6)$ & $n=2000,\ b=400$ & 0.024 & 0.017 & 0.031 & 0.029 & 94.8 & 94.2 & 94.6 & 95.4 & 94.4 & 95.4 \\
$(1/2,\,1/3,\,1/6)$ & $n=4000,\ b=800$ & 0.012 & 0.009 & 0.015 & 0.014 & 94.8 & 96.4 & 94.8 & 95.4 & 95.8 & 95.0 \\
\hline
\end{tabular}
}
\end{table}

\begin{table}[H]
\caption{\small Complete results of ADL-LTMLE using the misspecified working outcome regression $
Y \sim L_1 + A_1 + A_2 + A_3,
$ instead of the correct outcome model ($K=3$). The settings vary total sample size $n$ and cohort size $b$.}
\label{tab:k3_ltmle_sample_size_misspecification_sensitivity}\par
\centering
\resizebox{0.98\linewidth}{!}{
\begin{tabular}{|llrrrrrrrrrr|}
\hline
$w$ & Setting & TarVar & Var$_1$ & Var$_2$ & Var$_3$ & Cov$_1$(\%) & Cov$_2$(\%) & Cov$_3$(\%) & OCov$_1$(\%) & OCov$_2$(\%) & OCov$_3$(\%) \\
\hline
$(1/3,\,1/3,\,1/3)$ & $n=1000,\ b=200$ & 0.065 & 0.048 & 0.069 & 0.078 & 94.0 & 95.6 & 94.6 & 94.8 & 96.0 & 95.8 \\
$(1/3,\,1/3,\,1/3)$ & $n=2000,\ b=200$ & 0.032 & 0.024 & 0.032 & 0.040 & 94.6 & 94.8 & 94.2 & 95.4 & 95.4 & 94.6 \\
$(1/3,\,1/3,\,1/3)$ & $n=2000,\ b=400$ & 0.034 & 0.025 & 0.036 & 0.041 & 94.2 & 94.8 & 93.0 & 94.4 & 95.4 & 94.0 \\
$(1/3,\,1/3,\,1/3)$ & $n=4000,\ b=800$ & 0.017 & 0.013 & 0.018 & 0.019 & 93.0 & 94.2 & 95.8 & 95.6 & 95.4 & 96.2 \\
\hline
$(1/6,\,1/3,\,1/2)$ & $n=1000,\ b=200$ & 0.066 & 0.057 & 0.068 & 0.067 & 95.6 & 95.2 & 95.0 & 95.6 & 95.6 & 94.6 \\
$(1/6,\,1/3,\,1/2)$ & $n=2000,\ b=200$ & 0.032 & 0.029 & 0.032 & 0.033 & 93.8 & 94.6 & 94.4 & 94.6 & 95.0 & 94.8 \\
$(1/6,\,1/3,\,1/2)$ & $n=2000,\ b=400$ & 0.034 & 0.031 & 0.035 & 0.034 & 93.0 & 94.8 & 93.6 & 94.2 & 95.4 & 94.0 \\
$(1/6,\,1/3,\,1/2)$ & $n=4000,\ b=800$ & 0.017 & 0.015 & 0.017 & 0.017 & 93.0 & 95.4 & 96.2 & 94.0 & 96.2 & 96.2 \\
\hline
$(1/2,\,1/3,\,1/6)$ & $n=1000,\ b=200$ & 0.056 & 0.039 & 0.068 & 0.084 & 96.0 & 96.8 & 94.6 & 96.0 & 96.2 & 95.0 \\
$(1/2,\,1/3,\,1/6)$ & $n=2000,\ b=200$ & 0.028 & 0.019 & 0.031 & 0.046 & 95.0 & 95.6 & 94.4 & 94.8 & 95.0 & 94.4 \\
$(1/2,\,1/3,\,1/6)$ & $n=2000,\ b=400$ & 0.031 & 0.021 & 0.037 & 0.047 & 94.4 & 93.8 & 93.6 & 95.2 & 94.2 & 93.8 \\
$(1/2,\,1/3,\,1/6)$ & $n=4000,\ b=800$ & 0.015 & 0.011 & 0.018 & 0.022 & 93.6 & 95.6 & 94.6 & 93.8 & 95.4 & 95.0 \\
\hline
\end{tabular}
}
\end{table}